\documentclass[10pt]{article}
\usepackage{geometry}
\usepackage{enumerate}
\usepackage{latexsym,booktabs}
\usepackage{amsmath,amssymb}
\usepackage{graphicx}
\usepackage[onehalfspacing]{setspace}
\usepackage{color}
\usepackage{amsmath}
\usepackage{scalefnt}
\usepackage{amsthm}
\usepackage{mathrsfs}
\usepackage{multirow}
\usepackage{fancyvrb}
\usepackage{comment}
\usepackage{framed}
\usepackage{caption}
\usepackage{multicol}
\usepackage{multirow}
\usepackage{eurosym}
\usepackage{amsopn}
\usepackage{amssymb}
\usepackage[dvipsnames]{xcolor}
\usepackage{array}
\usepackage{amsthm}
\usepackage[toc,page]{appendix}
\usepackage{wrapfig}
\usepackage{subcaption}
\usepackage{indentfirst}
\usepackage{physics}
\usepackage{rotating}
\usepackage{soul}
\usepackage{color}
\usepackage{capt-of}
\usepackage{listings}
\usepackage{natbib}
\usepackage{pdflscape}
\usepackage{mathtools}
\usepackage{algorithm}

\newcolumntype{C}{>{$\displaystyle} c <{$}}
\newcolumntype{L}{>{\arraybackslash}m{8cm}}

\newtheorem{Theorem}{Theorem}
\newtheorem{Lemma}{Lemma}
\newtheorem{Proposition}{Proposition}
\newtheorem{Corollary}{Corollary}

\numberwithin{Theorem}{section}
\numberwithin{Corollary}{section}
\numberwithin{Definition}{section}
\numberwithin{Lemma}{section}
\numberwithin{Proposition}{section}
\numberwithin{Algorithm}{section}
\numberwithin{equation}{section}

\newcommand{\defeq}{\vcentcolon=}

\providecommand{\keywords}[1]{\textbf{Keywords:} #1}

\title{\vspace{-40pt}A characterisation of the  reconstructed birth-death process  through time rescaling}
\author{Anastasia Ignatieva\,\thanks{anastasia.ignatieva@warwick.ac.uk} \footnotemark[2] \footnotemark[3] ,  Jotun Hein\,\thanks{Department of Statistics, University of Oxford, 24-29 St Giles', Oxford OX1 3LB, UK} \footnotemark[5] , Paul A.\ Jenkins\,\thanks{Department of Statistics, University of Warwick, Coventry CV4 7AL, UK} \thanks{Department of Computer Science, University of Warwick, Coventry CV4 7AL, UK} \thanks{The Alan Turing Institute, British Library, London NW1 2DB, UK}}
\date{May 6, 2020}

\begin{document}

\captionsetup{width=0.8\textwidth}


\maketitle

\singlespacing
\begin{abstract}
The dynamics of a population exhibiting exponential growth can be modelled as a birth-death process, which naturally captures the stochastic variation in population size over time.  In this article, we consider a supercritical birth-death process, started at a random time in the past, and conditioned to have $n$ sampled individuals at the present.  The genealogy of individuals sampled at the present time is then described by the reversed reconstructed process (RRP), which traces the ancestry of the sample backwards from the present. We show that a simple, analytic, time rescaling of the RRP provides a straightforward way to derive its inter-event times. The same rescaling characterises other distributions underlying this process, obtained elsewhere in the literature via more cumbersome calculations. We also consider the case of incomplete sampling of the population, in which each leaf of the genealogy is retained with an independent Bernoulli trial with probability $\psi$, and we show that corresponding results for Bernoulli-sampled RRPs can be derived using time rescaling, for any values of the underlying parameters.  A central result is the derivation of a scaling limit as $\psi$ approaches 0, corresponding to the underlying population growing to infinity, using the time rescaling formalism. We show that in this setting, after a linear time rescaling, the event times are the order statistics of $n$ logistic random variables with mode $\log(1/\psi)$; moreover, we show that the inter-event times are approximately exponentially distributed. 
\end{abstract}

\hspace{7pt} \keywords{birth-death, reconstructed process, Bernoulli sampling, time rescaling}


\onehalfspacing

\section{Introduction}

The coalescent is a widely used model describing the genealogy of a sample taken from a population, arising as the scaling limit of numerous population models \citep{hei:etal:2005}. A key assumption of the basic coalescent is that the population size is large but constant or deterministically changing through time, although there are stochastic formulations \citep{kaj:kro:2003, par:etal:2010}. For some species, the dynamics of a population where individuals replicate and die independently of each other may be more naturally modelled as a birth-death process, which captures the stochastic nature and rapid growth of the population size \citep{boskova, stadlercoalescent}. The simple linear birth-death process (BDP) studied by \citet{kendalla} is a popular neutral population model, in which individuals independently divide at rate $\lambda$ and die at rate $\mu$. A realisation of this process can be represented as a tree relating the individuals, with bifurcations corresponding to birth events, and terminating branches corresponding to death events. The process models the entire population, creating a birth-death tree such as that shown on the left of Figure \ref{bd_trees}, where lineages can go extinct before the present. The genealogy of surviving individuals can then be obtained by pruning these extinct lineages, shown in the middle panel. The process tracing out the genealogy is termed the \emph{reconstructed process (RP)}  \citep{nee}. 

\citet{gernhard} considered the RP backwards in time, conditioning it on having $n$ extant individuals at the present and a given time of origin  $T$. \citet{gernhard} noted a correspondence between this conditioned reconstructed process and a point process with i.i.d.\ speciation times; this is termed a \emph{coalescent point process (CPP)} as introduced by \citet{aldouspopovic} for critical branching processes. With this formulation, and using the results of \citet{thompson}, \citet{gernhard} then derived the density of bifurcation times in the RP, conditioned on $T$. Then, using an improper uniform $(0, \infty)$ prior on $T$ and integrating, \citet{gernhard} obtained an expression for the density of the $k$-th bifurcation time.  In this article, we consider the time to origin to be random, similarly assuming a uniform prior on $T$, and condition on the sample size $n$ at the present. 

\begin{figure}[htbp!]
\centering
\includegraphics[height=7cm]{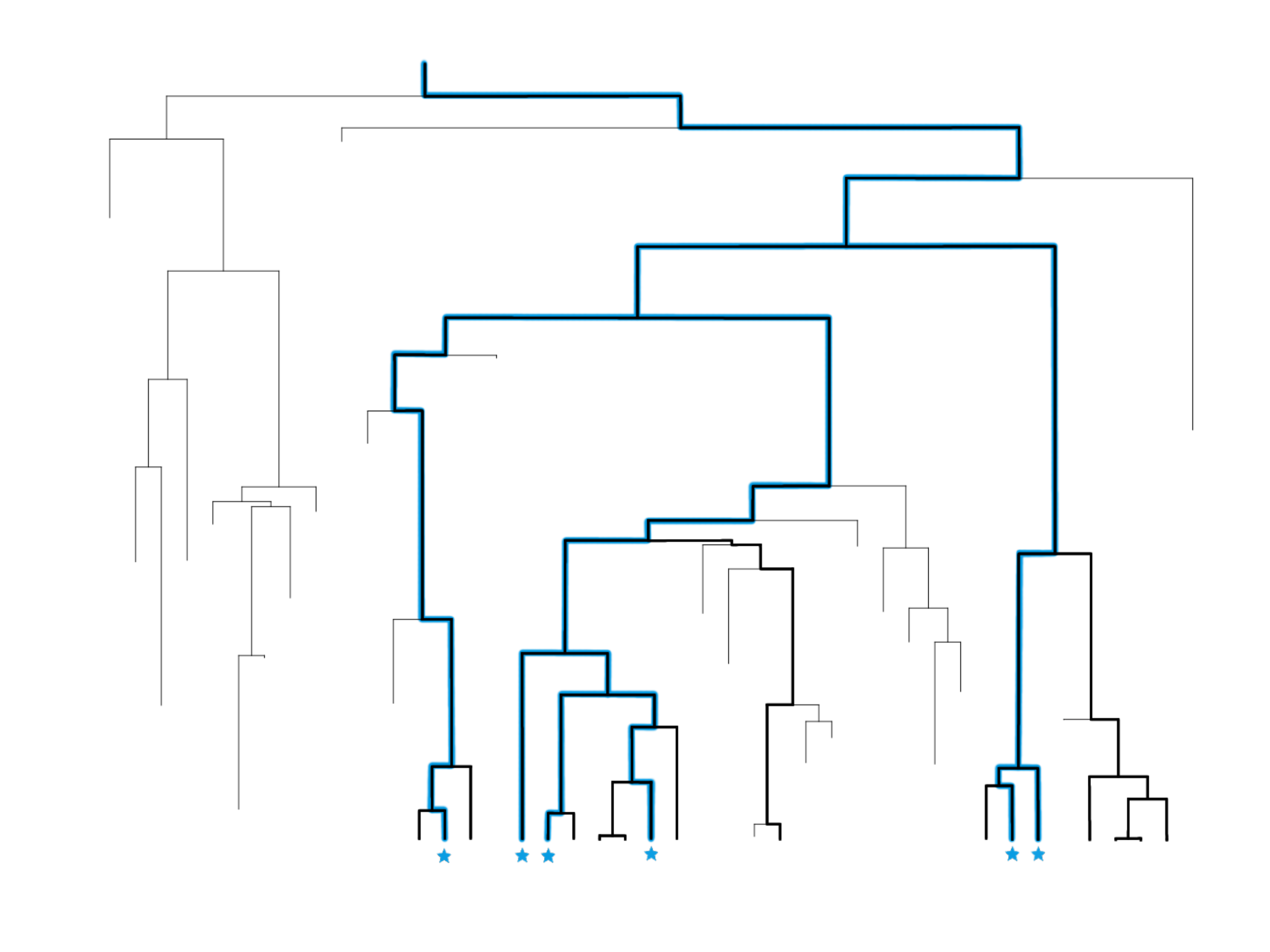}
\includegraphics[height=7cm]{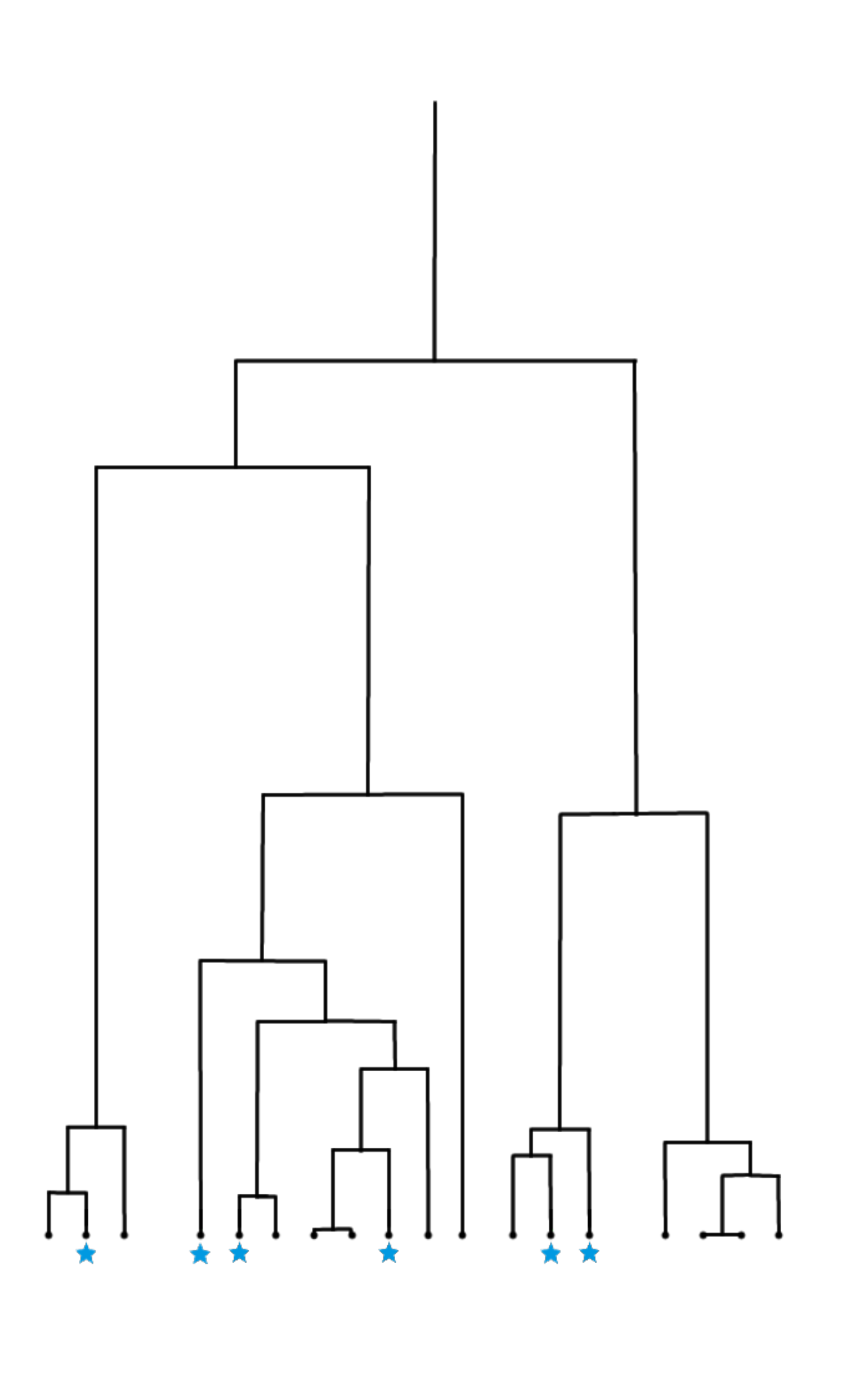}
\includegraphics[height=7cm]{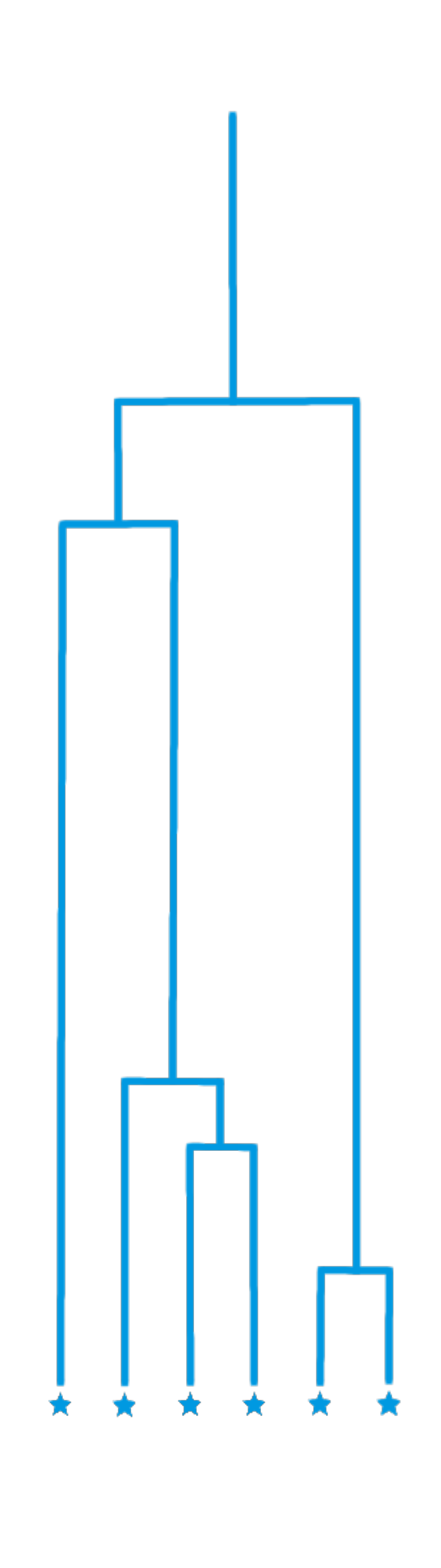}
\caption{Left: birth-death tree with $\lambda = 0.1, \mu = 0.05$ and 18 individuals surviving to present time. Middle: corresponding RRP with complete sampling. Right: RRP with incomplete sampling: each surviving individual is sampled independently with fixed probability $\psi = \frac{1}{3}$; blue stars indicate sampled individuals} \label{bd_trees}
\end{figure}

Birth-death models differ from the coalescent in that they must explicitly incorporate the sampling regime used in obtaining the sample. Two main sampling regimes have been considered in the literature: Bernoulli-type sampling (where each extant individual is sampled independently with some fixed probability $\psi$), and $n$-sampling ($n$ individuals are sampled from the full population, of size $N$ conditioned to be greater than $n$). \citet{stadlersampling} analysed the conditioned reconstructed process with Bernoulli sampling and derived the joint density of bifurcation times for the sample (conditioned on the time of origin, or with a uniform prior). \citet{wiuf_bd} and \citet{stadlerswapping} further looked at the correspondence between a complete and incompletely sampled process, by transforming the parameters. \citet{lambert} showed that there is a relationship between CPPs with Bernoulli sampling and CPPs with $n$-sampling: to simulate a CPP with $n$-sampling, one can first draw a realisation $y$ of a random variable with a specific density, and then simulate a CPP under Bernoulli sampling with $y$ as the sampling probability.  We focus on Bernoulli-type sampling for fixed values of $\psi$, and then consider the limit $\psi \to 0$, which corresponds to the underlying population size (in the complete tree) growing to infinity. 

 There is also a substantial related body of work concerning Bienaym\'e-Galton-Watson (BGW) processes, considered in either discrete or continuous time, in which individuals reproduce independently according to a specified offspring distribution (in continuous time, when the number of offspring is either zero or two, we obtain the special case of a birth-death process). Work on reduced trees (tracing the genealogy of a sample) goes back several decades; \citet{fleischmann} showed that the reduced tree associated with a BGW process is itself a time-inhomogeneous BGW process, similar to the result of \citet{nee} for the reconstructed birth-death process. Several papers have considered the question of coalescence times for a finite sample \citep[e.g.][]{oconnell,harris,grosjean,burden}. Similar to our treatment is that of \citet{oconnell}, who derives an expression for the coalescence time of a sample of size 2, as a fraction of the time since origin of the process; \citet{harris} generalise these results to any sample size, and consider continuous-time BGW processes sampled after time $T$, with $n$-sampling. Our work differs from the latter in the treatment of time to origin -- rather than assuming that sampling happens a fixed time after the origin of the process, we treat the time of origin as random. Moreover, we consider Bernoulli type sampling, so the connection between our results is limited to the setting of taking the sampling probability to 0, which is conceptually similar to taking $T$ to infinity. We focus specifically on results for birth-death processes, with any sample size, while \citet{harris} limit the exposition of their results explicitly applied to birth-death processes in the limit $T \to \infty$ to $n=2$.   Our treatment of time to origin is more similar to the work of \citet{burden}, who consider the infinite-population limit of near-critical Galton-Watson process, arriving at the Feller diffusion; the time of most recent common ancestor is treated as random. However, we do not consider the diffusion approximation, and focus on supercritical rather than near-critical processes. 


\subsection{The birth-death process (BDP)}

Consider a birth-death process $\mathcal{B}$ with birth rate $\lambda$ and death rate $\mu \neq \lambda$ (which we will shorten as BDP($\lambda, \mu$)), with time measured in units $\beta$. The process starts with one individual at time 0 and is run for time $T$ since origin, at which point all $n$ extant individuals are sampled. We assume a uniform (improper) prior on $T$, reiterating that this choice of prior is not novel, and has been treated, for instance, in \citet{aldouspopovic} for the critical case $\lambda = \mu$, and \citet{gernhard} and \citet{wiuf_bd} for the supercritical case. In this section, we give calculations showing that with this choice of improper prior, after conditioning on the number of sampled individuals $n$, the time since origin $T_n^*$ of the conditioned process is random with a particular, proper distribution. We then show that the BDP($\lambda, \mu$) sampled a time $T_n^*$ since origin and conditioned to have $n$ sampled individuals is dual to the BDP($\mu, \lambda$), initialised with $n$ individuals and run until first hitting state $0$. This is not a new result, but it is crucial to the idea of considering the reconstructed process backwards in time from sampling, so we include it for completeness.

\subsubsection{Prior on the time of origin}
Let ${\mathcal{B}}_n$ denote the process ${\mathcal{B}}$ conditioned to have $n$ sampled individuals, and denote by ${\mathcal{B}}_{n,t}$ this process with the sampling step happening at time $t$ since origin. Here we will consider both the subcritical ($\mu > \lambda$) and supercritical ($\lambda > \mu$) cases. Let $N(t)$ denote the number of individuals alive in  $\mathcal{B}_n$ at time $t$ since origin. Then the generating function of $N(t)$ is given by \citep[Chapter III, Section 5]{athreya_ney}:
\[
G(s) = \mathbb{E}(s^{N(t)}) = \frac{\mu (s-1) e^{(\lambda - \mu) t} - \lambda s + \mu}{\lambda (s-1) e^{(\lambda - \mu)t} - \lambda s + \mu}.
\]
Then
\[
p_t \defeq \mathbb{P} (N(t) = 0) = G(0) = \frac{\frac{\mu}{\lambda - \mu} (e^{(\lambda - \mu)t} - 1)}{1 + \frac{\lambda}{\lambda - \mu} (e^{(\lambda - \mu)t} - 1)},
\]
and 
\[
\mathbb{P} (N(t) = j) = (1-p_t) \left( 1 - \frac{\lambda}{\mu} p_t \right)p_t^{j-1}.
\]
Analogously to \citet[Section 2]{aldouspopovic}, define the probability measure
\begin{equation} \label{p_measure}
\mathbb{P}^* (\mathcal{B}_n \in \cdot) \defeq \frac{\int_0^\infty \mathbb{P}(\mathcal{B}_{n,t} \in \cdot) \mathbb{P}(N(t) = n) dt}{\int_0^\infty \mathbb{P}(N(t) = n) dt}.
\end{equation}
Making the observation that $\mathbb{P}(N(t) = j) = \frac{1}{\mu} \left(\frac{d}{dt} p_t \right) (p_t)^{j-1} = \frac{1}{j \mu} \frac{d}{dt}(p_t^j)$, we obtain
\begin{equation*}
\int_0^\infty \mathbb{P}(N(t) = n) \, dt = \frac{1}{n \mu} \left[ p_t^n \right]_0^\infty  =
\begin{cases}
\frac{1}{n \mu} &\text{if $\mu > \lambda$} \\
\frac{1}{n \mu} \left( \frac{\mu}{\lambda} \right) ^n &\text{if $\lambda > \mu$}.
\end{cases}
\end{equation*}
Then the function
\begin{equation} \label{origin_lambda_mu}
f_{T_n^*}(t) = \frac{\mathbb{P}(N(t) = n)}{\int_0^\infty \mathbb{P}(N(t) = n) dt} = 
\begin{cases}
\frac{n \mu e^{(\mu - \lambda)t} \left[ \frac{\mu}{\mu - \lambda} (e^{(\mu - \lambda)t} - 1) \right]^{n-1}}{\left[ 1 + \frac{\mu}{\mu - \lambda} (e^{(\mu - \lambda)t} - 1)\right]^{n+1}} &\text{if $\mu > \lambda$} \\[3ex]
\frac{n \lambda e^{(\lambda - \mu)t} \left[ \frac{\lambda}{\lambda - \mu} (e^{(\lambda - \mu)t} - 1) \right]^{n-1}}{\left[ 1 + \frac{\lambda}{\lambda - \mu} (e^{(\lambda - \mu)t} - 1)\right]^{n+1}} & \text{if $\lambda > \mu$}
\end{cases}
\end{equation}
is a probability density on $[0, \infty)$ for the time since origin $T_n^*$, and \eqref{p_measure} can be rewritten as
\[
\mathbb{P}^* (\mathcal{B}_n \in \cdot) = \int_0^\infty f_{T_n^*}(t) \; \mathbb{P}(\mathcal{B}_{n,t} \in \cdot) dt.
\]
Thus, after conditioning on the sample size $n$, $f_{T_n^*}(t)$ is a proper density for $T_n^*$.

\subsubsection{Time reversal} \label{bdp_reversal}

The population size of the BDP$(\lambda, \mu)$ is a continuous-time Markov chain with the transition rates
\[
q_{i, i+1} = \lambda i, \;\; q_{i, i-1} = \mu i.
\]
As above, denote by $\{ N(t), \; 0 \leq t \leq T_n^* \}$ the corresponding process associated with the complete tree, counting the population size up to the time of sampling, making the jump from 0 to 1 at time $0$. Consider also the continuous-time Markov chain $\{ \widehat{N}_n(s), \; 0 \leq s \leq T_n\}$, which has the reversed transition rates
\[
q_{i, i+1} = \mu i, \;\; q_{i, i-1} = \lambda i,
\]
started in state $\widehat{N}_n(0) = n$ and run until the first hitting time $T_n$ of state 0. Then, mirroring \citet[Lemma 2]{aldouspopovic} for the critical case, we have the following:
\begin{Lemma} \label{time_reverse_lemma}
\[
\{ N(T_n^* - s), \; T_n^* \geq s \geq 0\} \stackrel{d}{=} \{ \widehat{N}_n(s), \; 0 \leq s \leq T_n\},
\]
and in particular $T_n^* \stackrel{d}{=} T_n$, where $\stackrel{d}{=}$ denotes equality in distribution.
\end{Lemma}

\begin{proof}
Fix event times $s_0, \hdots, s_M$, with $s_M > s_{M-1} > \hdots s_1 > s_0 = 0$, and positive integers $k_M = 1, k_{M-1}, \hdots, k_2, k_1 = n$, with $|k_m - k_{m-1} = 1|$ and setting $k_{M+1} = 0$. The sequence of $k_m$'s describes a population size trajectory of a realisation of the birth-death process; reading from left to right, this has $b+1$ increases of size 1, and $n-1+b$ decreases of size 1 for some integer $b \geq 0$ with $n+2b = M$. Then the event 
\begin{quote}
\{as $s$ decreases, $N(T_n^*-s)$ jumps from $k_{m+1}$ to $k_m$ for $s \in [s_m, s_m + ds_m]$ (for all $M \geq m \geq 1$) and makes no other jumps\}
\end{quote}
has measure
\begin{equation} \label{meas_1}
ds_M \cdot \prod_{m=M}^2 \left( e^{-k_m(\lambda + \mu)(s_m-s_{m-1})} \; k_m \, ds_{m-1} \right) \cdot  \lambda^{b+1} \; \mu^{n-1+b}\cdot e^{-k_1s_1},
\end{equation}
where $ds_M$ comes from the uniform prior. For the reversed process $\widehat{N}_n(s)$, the event
\begin{quote}
\{as $s$ increases, $\widehat{N}_n(s)$ jumps from $k_{m}$ to $k_{m+1}$ in the interval $s \in [s_m, s_m + ds_m]$ (for all $1 \leq m \leq M$) and makes no other jumps\}
\end{quote}
has probability
\begin{equation} \label{meas_2}
\prod_{m=1}^M \left( e^{-k_m(\lambda + \mu)(s_m-s_{m-1})} \; k_m ds_{m} \right) \cdot \mu^{n-1+b} \;  \lambda^{b+1},
\end{equation}
as reading the sequence of $k_m$'s from right to left, there are $n-1+b$ increases of size 1 and $b+1$ decreases of size 1. The measure \eqref{meas_1} is $1/k_1 = 1/n$ times \eqref{meas_2}, so after conditioning the probability measures of the two events are equal. 
\end{proof}

This demonstrates the duality between the BDP$(\lambda, \mu)$ started with 1 individual at time 0 and reaching $n$ individuals after the random time $T_n^*$ (running ``forwards" to the time of sampling), and the BDP$(\mu, \lambda)$, started from $n$ individuals at time 0  and run until it reaches state 0 (running ``backwards in time" from the sample). We next consider the reconstructed process, which tracks the genealogy of only the sampled individuals, and make use of the duality between the forwards-in-time and backwards-in-time formulations.

\subsection{The reversed reconstructed process (RRP)} \label{rrp_intro}

 The RP (forwards in time) describes the number of lineages in the BDP($\lambda, \mu$), which will have at least one surviving descendant in the sample.  \citet{nee} identified that the RP forwards in time is generated by an underlying time-inhomogeneous pure birth process, with birth rate per lineage at time $t$ given by:
\begin{align}
\lambda P_1(t, T) &\defeq \lambda \cdot \mathbb{P}\text{(a single lineage born at time $t$ is not extinct by time $T$)} \nonumber \\
&= \frac{\lambda (\lambda - \mu)}{\lambda - \mu e^{-(\lambda - \mu)(T-t)}}  \nonumber \\
&= \frac{\lambda e^{(\lambda - \mu)(T-t)}}{1 + \frac{\lambda}{\lambda - \mu} (e^{(\lambda - \mu)(T-t)}-1)}, \label{nee_rate}
\end{align}
where $T$ is the time of sampling and $P_1(t, T)$ is given by \citet{kendalla}.  The state of the process at time $t$ is the number of individuals alive at $t$ with at least one descendant at the time of sampling $T$, with events corresponding to transitions from state $j$ to $j+1$, $j \geq1$.

It is advantageous to consider the process running backwards in time from the present, conditioning on the sample size $n$, and not explicitly conditioning on the time of origin of the process (which is generally unknown,  and for which we impose a uniform improper prior). In line with the discussion in Section \ref{bdp_reversal}, we will thus consider the properties of the \emph{reversed reconstructed process (RRP)}, which is defined as the process tracking the genealogy of the initial population of the BDP($\mu, \lambda$), initialised at $n$ individuals and run until the first hitting time of state 0. It is straightforward to show, similarly to Lemma \ref{time_reverse_lemma}, that the RP with birth rate \eqref{nee_rate} run for time $T_n^*$ and reaching state $n$ at the time of sampling is dual to the RRP which is started in state $n$ at time 0 and stopped at the first hitting time of state 0, with death rate obtained by replacing $T-\tau$ by $\tau$ in \eqref{nee_rate} to account for the time reversal. Note that the time index $\tau$ increases into the past, and $\tau=0$ denotes the time of sampling. The RRP is thus an inhomogeneous pure-death process, with death rate per lineage given by:
\begin{equation} \label{BDP1_m}
m_\beta(\tau) = \frac{\lambda e^{(\lambda - \mu)\tau}}{1 + \frac{\lambda}{\lambda - \mu} (e^{(\lambda - \mu)\tau}-1)}.
\end{equation}

To obtain the death rate of the RRP with Bernoulli sampling (where each lineage is sampled with a fixed probability $\psi$ at time 0), measured in time units of $\gamma$, we replace $P_1(\tau, T)$ with the relevant probability  $P_\psi(t, T)$ as derived by \citet{yangrannala}:
\begin{equation*}
P_\psi(t, T) = \frac{\psi (\lambda - \mu)}{\psi \lambda - (\mu - (1 - \psi) \lambda) e^{-(\lambda - \mu)(T-t)}} = \frac{\psi e^{(\lambda - \mu)(T-t)}}{1 + \frac{\psi \lambda}{\lambda - \mu} \Big( e^{(\lambda - \mu)(T-t)} - 1 \Big)},
\end{equation*}
which,  following the same reasoning as for the case of complete sampling,  gives the RRP death rate:
\begin{equation} \label{RRP_rate}
m_\gamma(\tau) = \frac{\psi\lambda e^{(\lambda - \mu)\tau}}{1 + \frac{\psi\lambda}{\lambda - \mu} (e^{(\lambda - \mu)\tau}-1)}.
\end{equation}

Note that for the case of a subcritical process (with $\lambda < \mu$), the population process backwards in time is supercritical. To ensure that the population reaches a common ancestor, we thus need to condition this process on ultimate extinction; it can be shown that this is equivalent to swapping the birth and death rate \citep{waugh},  indeed this is clear from \eqref{origin_lambda_mu} for the time of origin.  Thus, the RRP death rate in the subcritical case will be the same as \eqref{RRP_rate} but with $\lambda$ and $\mu$ interchanged. 

For the case of a critical branching process, measured in time units of $\alpha$, the death rate is given by taking the limit $\lambda \to \mu$ in \eqref{RRP_rate}:
\begin{equation*}
m_\alpha(\tau) = \frac{\psi \lambda}{1 + \psi \lambda \tau}.
\end{equation*}

\subsection{Overview}

In this paper we consider the RRP as a backwards in time inhomogeneous pure-death process, as described in Section \ref{rrp_intro} above. We show that properties of the RRP are easily derived using this formulation. We use this to re-derive several results, such as densities of event times, which have been given elsewhere in the literature, but stress that the resulting proofs are significantly simpler and more intuitive. 

Noting that there is a time rescaling between a  time-reversed  Yule rate 1 process and the RRP of the birth-death population model, we propose a new simulation algorithm for (incompletely) sampled RRPs using time rescaling. This is an alternative to existing algorithms \citep{hartmannsimulating, stadlersimulating}, which instead utilise a coalescent point process (CPP) formulation. We discuss the relationship between these two approaches. Further, we demonstrate the relationship between completely and incompletely sampled RRPs through time rescaling. In related work, e.g.\ \citet{stadlersteel}, the approach taken of transforming birth and death rates meant that results could be derived only for a restricted set of parameter values, in particular for $1-\psi \leq \mu/\lambda \leq 1$; this is especially restrictive when $\psi$ is small. Here, we  show instead that the completely and incompletely sampled RRPs are time-rescaled versions of each other, so distributions for the incompletely sampled case can be derived using a change of variables. We use this to complete the proof for the length of a randomly chosen pendant edge in \citet{stadlersteel} for all parameter values.

Next we consider the scenario in which the underlying population size in a birth-death process grows to infinity, but a finite sample of size $n$ is obtained. This can be thought of as taking the limit $\psi \to 0$ for the Bernoulli sampling probability; we discuss the connection with the limit as the total population size tends to infinity for $n$-sampling, using results of \citet{lambert}. We describe in detail the time transformation between the RRP in this setting to a  time-reversed  Yule rate 1 process; in this scenario, there are two distinct timescales, separating the time of the first event from the events nearer the root of the tree. The RRP tree becomes star-shaped: the terminal branch lengths tend to infinity, while the inter-event times at the top of the tree are approximately exponentially distributed on a shorter timescale, with rate depending on the remaining number of lineages. We then use the time rescaling formalism to derive, analytically, the density of the inter-event times, both for any $\psi \in (0,1]$ and in the limit $\psi \to 0$; both results are new to the best of our knowledge.  We then show that in the limit $\psi \to 0$, the event times are distributed as the order statistics of $n$ logistic random variables, with mode $\log(1/\psi)$ (after a simple, linear, time rescaling). Further, we show that the inter-event times (thus distributed as the spacings between consecutive order statistics of $n$ logistic random variables)  are approximately exponentially distributed, with error bounded by $1/n$ in terms of Kolmogorov-Smirnov distance. We also show that the expectation of inter-event times agrees exactly with the expectation under this approximation.

The paper is structured as follows. In Section \ref{notation} below, we introduce the notation used throughout. In Section \ref{theory}, we state several known results for inhomogeneous birth-death processes which we will rely on, and review time rescaling for these processes. In Section \ref{BDP_section}, we consider the RRP of birth-death processes with Bernoulli sampling. In Section \ref{psi_0_section}, we focus on the limit of the sampling probability going to 0. Finally, discussion is presented in Section \ref{discussion_section}. Proofs can be found in Appendix \ref{proofs_appendix}. 

Illustrations of trees throughout were made using the R package \texttt{ape} \citep{ape}.


\subsection{Notation} \label{notation}

Table \ref{table1} summarises the notation used throughout. For instance, $\text{BDP}(\lambda, \mu, \psi)$ denotes the birth-death population process where each individual divides independently with rate $\lambda$, dies independently with rate $\mu<\lambda$, with the rates measured in time units $\gamma$; at time 0, each surviving individual is sampled with a fixed probability $\psi$. The corresponding RRP, i.e.\ the process tracing out the genealogy of the sample from this population backwards in time from 0, is denoted by $X_\psi^\gamma \defeq (X_\psi^\gamma(\tau): \tau \geq 0)$.  We write $X_\psi^{\xi}$ to denote the same process, but with time rescaled to units of $\xi = g(\gamma)$ for some time transformation $g$, i.e.\  $X_\psi^\xi(g(\tau)) = X_\psi^\gamma(\tau)$.  We denote the death rates of $X_\psi^\gamma$ and $X_\psi^{\xi}$ by $m_\gamma$ and $m_\xi$, respectively, with the subscripts denoting the time scale on which the rates are measured.

\begin{table}[htbp!] 
\centering 
\begin{tabular}{L c l c} 
\hline
\bf{Population process}                                                      & \bf{Time unit}                & \bf{Notation}                        & \bf{RRP notation} \\ \hline
Yule process, 

birth rate $1$                              & $t$                         & $\text{Yule}(1)$           &  $Y$  \\ \hline
Critical branching process, 

birth = death rate $\lambda$, sampling probability $\psi$ & $\alpha$   & $\text{CBP}(\lambda, \psi)$      & $Z_\psi^\alpha$   \\ \hline
Birth-death process, 

birth rate $\lambda$, death rate $\mu$, complete sampling    & $\beta$       & $\text{BDP}(\lambda, \mu, 1)$    & $X_1^\beta$  \\ \hline
Birth-death process, 

birth rate $\lambda$, death rate $\mu$, sampling probability $\psi$ & $\gamma$ &$\text{BDP}(\lambda, \mu, \psi)$ & $X_\psi^\gamma$  \\ \hline

Birth-death process, 

birth rate $\lambda'$, death rate $\mu'$, with $\lambda' - \mu' = 1$

sampling probability $\psi$ & $\delta$ &$\text{BDP}(\lambda', \mu', \psi)$ & $X_\psi^\delta$  \\ \hline

\end{tabular}
\caption{Summary of notation}\label{table1}
\end{table}

A table summarising the properties of each RRP is given in Appendix \ref{appB} for reference.


\section{Background} \label{theory}

We briefly review relevant known results which we will rely on throughout the paper.

\subsection{Inhomogeneous  pure-death  processes}
We briefly state relevant known results concerning inhomogeneous  pure-death  processes. Consider a time-inhomogeneous pure-death process, with time measured in units $\xi$, starting with $n$ individuals alive at time $0$. Each individual dies independently at rate $m_\xi(\tau)$; if there are $j$ individuals at time $\tau$, the intensity is $j m_\xi(\tau)$. The rate function of the process is given by 
\begin{equation*} \label{rho}
\rho_\xi(\tau) = \int_0^\tau m_\xi(x) dx.
\end{equation*}
The transition probabilities, i.e.\ the probability of going from $n$ to $j$ individuals in time $\tau$, are given by a binomial distribution \citep[p.112]{bailey}:
\begin{equation} \label{transitions}
P_{nj}(\tau) = 
\begin{cases}
{n \choose j} \big(1 - e^{- \rho_\xi(\tau)} \big)^{n-j} \big(e^{{-\rho_\xi(\tau)}}\big)^j &\text{for $j \leq n$},\\
0 &\text{ otherwise,}
\end{cases}
\end{equation}
with $e^{-\rho_\xi(\tau)}$ being the probability that a lineage has not died by time $\tau$. The distribution of time to origin is \citep[p.112]{bailey}:
\begin{equation} \label{origin_cdf}
F_{T_n}(\tau) = P(T_{n} < \tau) = \big( 1 - e^{-\rho_\xi(\tau)} \big) ^n,
\end{equation}
and, by differentiating, the pdf is
\begin{align} \label{origin_pdf}
f_{T_n}(\tau) &=  n m_\xi(\tau) e^{-\rho_\xi(\tau)}\big( 1 - e^{-\rho_\xi(\tau)} \big) ^{n-1}.
\end{align}
The density of the time of the $k$-th event is given by
\begin{align}
f_{T_k}(\tau) &= {n \choose k} \cdot \underbrace{k m_\xi(\tau) e^{-\rho_\xi(\tau)} \big(1 - e^{-\rho_\xi(\tau)}\big)^{k-1}}_{\text{$k$-th lineage dies at $\tau$}} \;\;\; \cdot \underbrace{\big(e^{-\rho_\xi(\tau)}\big)^{n-k}}_{\text{$n-k$ survive for at least $\tau$}} \nonumber \\
&= {n \choose k} \;k \; m_\xi(\tau) \big( 1 - e^{-\rho_\xi(\tau)} \big)^{k-1}  \big(e^{-\rho_\xi(\tau)} \big)^{n-k+1}. \label{wait_pdf}
\end{align}


\subsection{Time rescaling} \label{timerescaling}

Consider a pure-death inhomogeneous process with death rate $m_\xi(\tau)$, with time measured in units of $\xi$. Suppose that time is rescaled in units of $\zeta = g(\xi)$, where $g$ is strictly monotonic and differentiable. The death rate of the process then becomes, using a change of variables:
\begin{equation*} \label{change_of_vars}
m_\zeta(\tau) = m_\xi(g^{-1}(\tau)) \abs{\frac{d}{d \tau} \, g^{-1}(\tau)}.
\end{equation*}

The time rescaling theorem, due to \citet{tr1} and \citet{tr2}, states that any inhomogeneous point process with an integrable intensity function can be rescaled to a Poisson process with unit rate. The RRP can be thought of as a point process, with intensity given by its inhomogeneous death rate times the number of lineages. If the RRP (of any population process) has death rate $m_\xi(\tau)$, then rescaling time via the transformation $g = \rho_\xi$ rescales the RRP to a homogeneous pure-death process with death rate per lineage equal 1 (a  time-reversed  Yule rate 1 process).


\subsection{ Time-reversed  Yule rate 1 process}

 We define the time-reversed Yule rate 1 process as  a pure death process where each lineage dies independently at rate 1,  denoted $Y$. This is the RRP of a forwards-in-time Yule process with birth rate 1.  The inter-event time  during which there are exactly $j$ lineages  is exponentially distributed with rate $j$. Using \eqref{origin_pdf}, the time to origin has density:
\begin{equation*}
f_{T_n}(\tau) = n e^{-\tau} (1 - e^{-\tau})^{n-1},
\end{equation*}
and using \eqref{wait_pdf}, the time to $k$-th event has density:
\begin{equation} \label{pd1_kthevent}
f_{T_k}(\tau) = {n \choose k} \; k \; \big(1 - e^{-\tau} \big)^{k-1} \big(e^{-\tau} \big)^{n-k+1}.
\end{equation}
The expectation of time to origin is $\sum_{j=1}^n \frac{1}{j}$. These results are identical to those derived by \citet{gernhardcbp}.


\section{Birth-death process with Bernoulli sampling} \label{BDP_section}

We now consider in detail the RRP $X_\psi^\gamma$ of a supercritical birth-death process. Using the formulation introduced in Section \ref{rrp_intro}, we first re-derive some known properties of the process, which will be readily available from the results given in Section \ref{theory}. Then, using the fact that the RRP $X_\psi^\gamma$ is a time rescaling of  the RRP associated with  a Yule rate 1 process, we propose a simulation algorithm. Finally, we discuss the relationship between completely and incompletely sampled RRPs through time rescaling. 


\subsection{Properties of the process}

Set $T_0 = 0$ and for $k \in \{ 1, \hdots, n \}$ denote by $T_k$ the time of the $k$-th event, backwards from the present time 0. At $T_k$, the number of lineages decreases from $n-k+1$ to $n-k$. For $k \in \{0, \hdots n-1 \}$, let $W_k \defeq T_{k+1} - T_k$ denote the inter-event time.

\subsubsection{Transition probabilities and densities of event times} \label{RRP_props}

We use the pure-death process formulation of $X_\psi^\gamma$ to derive distributions characterising this process. The transition probabilities are, using \eqref{transitions}:
\begin{equation*}
P_{ij}(\tau) = 
\begin{cases}
{i \choose j} \big(1 - e^{- \rho_\gamma(\tau)} \big)^{i-j} \big(e^{{-\rho_\gamma(\tau)}}\big)^j &\text{for $j \leq i$}\\
0 &\text{otherwise,}
\end{cases}
\end{equation*}
where, by integrating the death rate in \eqref{RRP_rate}:
\begin{equation} \label{rho_gamma}
\rho_\gamma(\tau) = \int^\tau_0 m_\gamma(x) dx = \log(1 + \frac{\psi \lambda}{\lambda - \mu} \Big( e^{(\lambda - \mu)\tau} - 1 \Big)),
\end{equation}
and
\begin{equation*}
e^{-\rho_\gamma(\tau)} = \frac{1}{1  + \frac{ \psi \lambda}{\lambda - \mu} (e^{(\lambda - \mu)\tau} - 1)}.
\end{equation*}
For $\tau \to \infty$ and fixed $\psi \in (0,1]$, we have $\rho_\gamma(\tau) \to \infty$ and $e^{-\rho_\gamma(\tau)}  \to 0$, so $P_{ij}(\tau) \to 0$ for all $j \neq 0$, and $P_{i0}(\tau) \to 1$.  This implies that two individuals sampled at the present will eventually find a common ancestor in the past with probability 1.

The distribution of time to origin, using \eqref{origin_cdf}, is given by:
\begin{align*}
F^\psi_{T_n}(\tau) = P(T_{n}< \tau) = \big( 1 - e^{-\rho_\gamma(\tau)} \big) ^n &=  \Bigg( \frac{\frac{ \psi \lambda}{\lambda - \mu} (e^{(\lambda - \mu)\tau} - 1)}{1  + \frac{ \psi \lambda}{\lambda - \mu} (e^{(\lambda - \mu)\tau} - 1)} \Bigg)^n,
\end{align*}
and its density, using \eqref{origin_pdf}, is
\begin{align}
f^\psi_{T_n}(\tau) &=  n \cdot \frac{\psi\lambda e^{(\lambda - \mu)\tau}}{1 + \frac{\psi\lambda}{\lambda - \mu} (e^{(\lambda - \mu)\tau}-1)} \cdot \frac{1}{1  + \frac{ \psi \lambda}{\lambda - \mu} (e^{(\lambda - \mu)\tau} - 1)}  \cdot \Bigg( \frac{\frac{ \psi \lambda}{\lambda - \mu} (e^{(\lambda - \mu)\tau} - 1)}{1  + \frac{ \psi \lambda}{\lambda - \mu} (e^{(\lambda - \mu)\tau} - 1)}  \Bigg)^{n-1} \nonumber \\
&=  n \psi \lambda e^{(\lambda - \mu)\tau} \frac{ \Big[ \frac{\psi \lambda}{\lambda - \mu}  \big(e^{(\lambda - \mu)\tau} - 1 \big) \Big]^{n-1} }{\Big[1 + \frac{\psi \lambda}{\lambda - \mu}\big( e^{(\lambda - \mu)\tau} - 1 \big)\Big]^{n+1}} \label{pdf_psi}.
\end{align}
 Note that this agrees with \eqref{origin_lambda_mu} for the case $\psi = 1$.  This result is also obtained in \citet[Lemma 3.1]{stadlersampling}. We note that although the outcome is identical, the derivation given above is significantly simpler, and follows directly from the properties of the RRP as a stochastic process. In particular, the distribution function is immediately obtained from knowing the death rate; moreover, to obtain the pdf we do not need to integrate over the prior for the time of origin, as this is implicit in the time reversal.

Using \eqref{wait_pdf}, the waiting time to the $k$-th event is given by:
\begin{align}
f^\psi_{T_k}(\tau) &= {n \choose k} \;k \; m_\gamma(\tau) \big( 1 - e^{-\rho_\gamma(\tau)} \big)^{k-1}  \big(e^{-\rho_\gamma(\tau)} \big)^{n-k+1} \nonumber \\
&= {n \choose k} 
\; k \; 
\frac{\psi\lambda e^{(\lambda - \mu)\tau}} {1 + \frac{\psi\lambda}{\lambda - \mu} (e^{(\lambda - \mu)\tau}-1)}
 \Bigg( 
 \frac{\frac{ \psi \lambda}{\lambda - \mu} (e^{(\lambda - \mu)\tau} - 1)}{1  + \frac{ \psi \lambda}{\lambda - \mu} (e^{(\lambda - \mu)\tau} - 1)}
 \Bigg)^{k-1} 
 \Bigg( 
\frac{1}{1  + \frac{ \psi \lambda}{\lambda - \mu} (e^{(\lambda - \mu)\tau} - 1)}
 \Bigg)^{n-k+1} \nonumber \\
&= {n \choose k} \; k \; \psi \lambda e^{(\lambda - \mu)\tau} \frac{ \Big[ \frac{\psi \lambda}{\lambda - \mu}  \big(e^{(\lambda - \mu)\tau} - 1 \big) \Big]^{k-1} }{\Big[ 1 + \frac{\psi \lambda}{\lambda - \mu}\big( e^{(\lambda - \mu)\tau} - 1 \big) \Big]^{n+1}}.  \label{kth_eventtime}
\end{align}
This agrees with the result derived in \citet[Theorem 4.1]{gernhard} for the case of complete sampling; we again note that the result follows almost immediately from the properties of the RRP, which removes the need for deriving the related distributions by hand. 

\subsubsection{Simulating from the RRP}
As described in Section \ref{timerescaling}, applying the time transformation $g_1 = \rho_\gamma$ rescales the RRP $X_\psi^\gamma$ to  the time-reversed Yule rate 1 process  $Y$. From \eqref{rho_gamma}, this transformation is given by:
\begin{align} 
&t = g_1(\gamma) = \log(1 + \frac{\psi \lambda}{\lambda - \mu} \Big( e^{(\lambda - \mu)\gamma} - 1 \Big)), \nonumber \\
&\gamma = g_1^{-1}(t) = \rho_\gamma^{-1}(t) = \frac{1}{\lambda - \mu} \log(1 + \frac{\lambda - \mu}{\psi \lambda} \Big( e^{t} - 1 \Big)), \label{t_to_gamma} 
\end{align}
and we have that 
\begin{equation*}
X_\psi^\gamma(g_1^{-1}(\tau)) = Y(\tau) \text{ and }
X_\psi^\gamma(\tau) = Y(g_1(\tau)),
\end{equation*} 
by which we mean that $X_\psi^\gamma$ rescaled in time units $g_1(\gamma)$ has the same death rate as $Y$. To see why this works, the death rate of $X_\psi^\gamma$ when measured in units $t = g_1(\gamma)$ becomes:
\[
m_t(\tau) = m_\gamma(g_1^{-1}(\tau)) \abs{\frac{d}{d\tau} \, g_1^{-1} (\tau)} = m_\gamma(\rho_\gamma^{-1}(\tau)) \abs{\frac{d}{d\tau} \, \rho^{-1}_\gamma(\tau)} =  m_\gamma(\rho_\gamma^{-1}(\tau)) \Big/m_\gamma(\rho_\gamma^{-1}(\tau)) = 1.
\]

 Note also that in the complete process, birth, death, and sampling events affect all individuals with equal probability, so the RRP trees we study have the same law in topology as Yule and coalescent trees \citep{aldous_clad}, and can thus be generated backwards in time by merging pairs of lineages selected uniformly at random.  This suggests that to simulate from $X_\psi^\gamma$, first we can simulate from  $Y$, and then rescale the event times using the transformation given by \eqref{t_to_gamma}. The method is summarised as Algorithm 1. This provides an alternative to the algorithms of \citet{hartmannsimulating} and \citet{stadlersimulating}, where first the time of origin is drawn from its distribution, and then the coalescent point process formulation is used to obtain the event times.

\begin{algorithm}[htbp!]
\caption{Simulating from the RRP $X_\psi^\gamma$} \label{alg1}
Given $n$ individuals at time 0:
\begin{enumerate}
\item Draw $\widetilde{W}_j \sim \text{Exp}(n-j) $ for $j = 0, \hdots, n-1$, being the waiting times  of $Y$. 
\item Compute the event times $\widetilde{T}_{j+1} = \sum_{i=0}^{j} \widetilde{W}_i$.
\item Rescale the event times as $T_k= \frac{1}{\lambda - \mu} \log(1 + \frac{\lambda - \mu}{\psi \lambda} \left( \exp(\widetilde{T}_k) - 1 \right) )$ for $k = 1, \hdots, n$.
\item Construct a tree from $T_1, \hdots, T_n$ by choosing a pair of lineages uniformly at random to coalesce at each event time.
\end{enumerate}
\end{algorithm}

Note that one can first derive distributions of interest for  $Y$, and then use the change of variables given by \eqref{t_to_gamma} to obtain the equivalent results for $X_\psi^\gamma$. We will use this to derive the distribution of inter-event times $W_k$, analytically, in Section \ref{interevent}. 

\subsubsection{Relationship with coalescent point processes} \label{CPP_link}

\citet{gernhard} gives the following CPP formulation for a supercritical process. To simulate an RRP for a sample of size $n$, first condition on the sample size and a time of origin $T_n$ (possibly drawn from the distribution \eqref{origin_cdf}), and then draw the times of the $n-1$ bifurcations in the tree i.i.d.\ from some specific density depending on $T_{n}$. \citet{lambertstadler} further give this density for the case of Bernoulli sampling. In a sense, conditioning on the time of origin, the event times can thus be simulated ``horizontally", one-by-one for each sampled lineage, rather than ``vertically", i.e.\ forwards or backwards in time. 

The formulation of the RRP as a pure-death process also allows for simulation of the RRP lineage-by-lineage, conditioning on the sample size but not on the time of origin (producing a tree including the root edge). Because each lineage dies independently from the others, in order to simulate from $X_\psi^\gamma$ for a sample of size $n$, we can simulate the death times of each of the $n$ lineages independently, and then merge the lineages uniformly at random at each event time to create the tree. The death time of one lineage has density:
\begin{align} \label{death_one}
f^{\psi}_{T_{(1)}}(\tau) =   \frac{\psi \lambda e^{(\lambda - \mu)\tau}}{\Big[1 + \frac{\psi \lambda}{\lambda - \mu}\big( e^{(\lambda - \mu)\tau} - 1 \big)\Big]^{2}},
\end{align}
which is obtained from \eqref{pdf_psi} by substituting $n=1$; this can be simulated by drawing from an exponential rate 1 density, and rescaling time using \eqref{t_to_gamma}.  Therefore the relationship between CPP and the pure-death formulation is very direct. With the pure-death formulation, each  of the $n$ lineages dies independently with the same death rate. Once we also condition on a time of origin $T_{n}$, the lineages still die independently, with death rate amended so that each event happens before $T_{n}$. The latter is exactly the CPP formulation of \citet{gernhard}. 

The CPP formulation described in \citet{lambertstadler} also gives a method for simulating a Bernoulli RRP without conditioning on the sample size, as follows. Given a time of origin $T$, draw realisations $H^\psi_1, \hdots, H^\psi_N$ of a random variable $H^\psi$, with the stopping criterion that $H^\psi_N$ is the first realisation that is greater than $T$. Then the $H^\psi_1, \hdots, H^\psi_{N-1}$ are the event times up to the MRCA for a sample of $N$ lineages in a Bernoulli sampled RRP, conditioned on time of origin $T$. Note that in this case, setting $p = P(H^\psi>T)$, the number of sampled lineages is geometric with mass function $(1-p)^{n-1}p$, and the density of $H^\psi$ given in \citet[p.122]{lambertstadler} is exactly that in \eqref{death_one}.

The pure-death formulation of the RRP highlights two differences between the genealogy of a birth-death process and the coalescent. Firstly, viewing the basic coalescent as a backwards in time pure-death process with rate $\frac{j(j-1)}{2}$ when there are $j$ lineages, at each point in time the death rate of each individual lineage depends on the total number of lineages remaining; this dependence cannot be removed by conditioning on the time of origin (for $n>2$). This implies that the process cannot be simulated by drawing the death time of each lineage independently from some density, as for the RRP. This supports the conjecture of \citet{lambertstadler} that the coalescent does not have a CPP representation.

Secondly, the coalescent with variable population size, as described by \citet{griffithstavare}, can be described as an inhomogeneous pure-death process, where the death rate is quadratic in the number of lineages and depends on a population size function. Because the death rate of the RRP is linear in the number of lineages, there is no population size function which would equate the two models.


\subsection{Relationship between completely and incompletely sampled RRPs}

\citet{stadlersampling} noted that there is a relationship between the RRP of the incompletely sampled BDP($\lambda, \mu, \psi$), and the RRP of the completely sampled BDP($\widehat\lambda, \widehat\mu, 1$), through the following transformation of the birth and death parameters:
\begin{align}
&\widehat\lambda = \psi \lambda, \;\; \widehat\mu = \mu - \lambda(1 - \psi). \label{sampling_trf}
\end{align}
Substituting \eqref{sampling_trf} as the birth and death rates into \eqref{BDP1_m} gives \eqref{RRP_rate}. Thus, the resulting process looks like the RRP of an incompletely sampled BDP$(\lambda, \mu,\psi$) population process. However, as noted by \citet{stadlersteel}, $\widehat \mu$ can be negative (in particular, for very small values of $\psi$); for instance, with the parameters used in Figure \ref{bd_trees}, $\widehat\mu = -1/60$. In this case, the interpretation as an RRP of some birth-death process is problematic. \citet{stadlersteel,stadlerswapping} discuss that when distributions are derived for the completely sampled process, this reparameterisation trick can be used to obtain the equivalent distributions for a process with incomplete sampling, but only for $\frac{\mu}{\lambda} \geq 1 - \psi$. Thus, this method of transforming the birth and death rates does not always produce a valid mapping between completely and incompletely sampled RRPs.

To avoid this issue, instead of transforming the birth and death parameters directly, we use a transformation of time, and demonstrate the relationship between the RRPs $X_\psi^\gamma$ and $X_1^\beta$. We do not introduce restrictions on the values of the parameters ($\lambda, \mu, \psi$), so this allows distributions derived for the completely sampled process to be transformed for the case of incomplete sampling. 

\subsubsection{Time transformation from $X_\psi$ to $X_1$}

Define the transformation of time units $g_2$ as:
\begin{align} \label{tprime}
\beta = g_2(\gamma) =  \frac{1}{\lambda - \mu} \log (1 + \psi (e^{(\lambda - \mu)\gamma} - 1)), \\
\gamma = g_2^{-1}(\beta) = \frac{1}{\lambda - \mu} \log (1 + \frac{1}{\psi} (e^{(\lambda - \mu)\beta} - 1)). \nonumber
\end{align}
This is a valid time transformation with $\gamma = 0 \iff \beta = 0$, and $\gamma = \beta$ when $\psi = 1$. Using a change of variable in \eqref{RRP_rate}, we compute the death rate:
\begin{align*}
m_\beta(\tau) &= m_\gamma(g_2^{-1}(\tau)) \cdot \abs{\frac{dg_2^{-1}(\tau)}{d\tau}} \\
&= \frac{\psi \lambda (1 + \frac{1}{\psi} (e^{(\lambda - \mu)\tau} -1))}{1 + \frac{\lambda}{\lambda - \mu} (e^{(\lambda - \mu)\tau} - 1)} \cdot \frac{ \frac{1}{\psi} e^{(\lambda - \mu)\tau}}{1 + \frac{1}{\psi} (e^{(\lambda - \mu)\tau}-1)} \\
&= \frac{\lambda e^{(\lambda - \mu)\tau}}{1 + \frac{\lambda}{\lambda - \mu} (e^{(\lambda - \mu)\tau} - 1)}.
\end{align*}
This is the death rate of the completely sampled RRP $X_1^\beta$ as given in \eqref{BDP1_m}.  Thus, we have the relationship:
\begin{align*}
&X_1^\beta(\tau) = X_\psi^\gamma (g_2^{-1}(\tau)), \\
 &X_1^\beta (g_2(\tau)) = X_\psi ^\gamma (\tau).
\end{align*} 
The RRP of a BDP$(\lambda, \mu, \psi)$ process is a time rescaled version of the RRP of a completely sampled BDP$(\lambda, \mu, 1)$ process. In effect, introducing incomplete sampling is equivalent to non-linearly rescaling the RRP of the BDP$(\lambda, \mu, 1)$ process using the time transformation \eqref{tprime}. 

\subsubsection{Deriving results for $X_\psi$ from $X_1$}

Using the time transformation approach, distributions can be derived for $X_1^\beta$ with complete sampling, and then the equivalent distribution results for $X_\psi^\gamma$ can be obtained through a simple change of variables. As an example, \citet{stadlersteel} derive the density of the length of a randomly chosen pendant edge (an edge adjacent to a sampled individual)   for an incompletely sampled tree with the restriction $1 - \psi \leq \frac{\mu}{\lambda} \leq 1$; we complete the proof for the case $0 \leq \frac{\mu}{\lambda} \leq 1 - \psi$. 

\begin{Proposition}
The density of a the length of a randomly chosen pendant edge, $E$, of the RRP $X_\psi^\gamma$ for any $0 \leq \mu < \lambda$ and $\psi \in (0,1]$ is
\begin{equation*}
f^\psi_E(\tau) = \frac{2 \psi \lambda (\lambda - \mu)^3 e^{(\lambda - \mu)\tau}}{\big(\lambda \psi e^{(\lambda - \mu)\tau} - [\mu - \lambda(1 - \psi)] \big)^3}.
\end{equation*}
\end{Proposition}

\begin{proof}
\citet{mooers} give the density of the length of a pendant edge of a completely sampled RRP $X_1^\beta$ as:
\begin{equation} \label{pendant}
f^1_E(\tau) =  \frac{2 \lambda (\lambda - \mu)^3 e^{(\lambda - \mu)\tau}}{(\lambda e^{(\lambda - \mu)\tau} - \mu)^3}.
\end{equation}
Using the time rescaling \eqref{tprime} and a change of variable, for $X_\psi^\gamma$ this becomes:
\begin{align*}
f^\psi_E(\tau) &= f^1_E(g_2(\tau)) \; \abs{\frac{d \, g_2(\tau)}{d\tau}}\\
&= \frac{2 \lambda (\lambda - \mu)^3 [1 + \psi (e^{(\lambda - \mu)\tau} - 1)]}{\big(\lambda \big[1 + \psi (e^{(\lambda - \mu)\tau} - 1) \big] - \mu \big)^3} \cdot \frac{\psi e^{(\lambda - \mu)\tau}}{1 + \psi (e^{(\lambda - \mu)\tau} - 1)}\\
&= \frac{2 \psi \lambda (\lambda - \mu)^3 e^{(\lambda - \mu)\tau}}{\big(\lambda \psi e^{(\lambda - \mu)\tau} - [\mu - \lambda(1 - \psi)] \big)^3}.
\end{align*} 
\end{proof}
Equivalence with the result of \citet[Section 4]{stadlersteel} for $\frac{\mu}{\lambda} \geq 1 - \psi$ is easily checked by substituting the birth rate $\widehat \lambda$ and death rate $\widehat \mu$ into \eqref{pendant}.


\section{Sampling from large populations} \label{psi_0_section}

We now consider the setting where the total population size is very large compared to the sample size $n$. This is a scenario often encountered in practice when collecting genetic data, particularly from viral or bacterial populations, when the population size is unknown but can be presumed very large. An example will be mentioned within the discussion in Section \ref{discussion_section}.

This situation is to be distinguished from the limit as the sample size grows to infinity, which has been considered in \citet{wiuf_bd}. The scenario of interest here is when the total population tends to infinity, but a finite sample of size $n$ is obtained. This can be interpreted as either the Bernoulli sampling probability $\psi$ going to 0, or the total population size growing to infinity in the case of $n$-sampling. In the following section we will discuss why the two regimes are conceptually similar.

In this  section, for the sake of readability of the expressions, we rescale time linearly in units of $\delta = (\lambda - \mu)\gamma$, and write $\lambda' = \frac{\lambda}{\lambda - \mu}, \mu' = \frac{\mu}{\lambda - \mu}$ with $\lambda' - \mu' = 1$. This simplifies the formulas, and is easy to reverse within any derived expressions. The RRP on this timescale is denoted $X_\psi^\delta$, with death rate 
\begin{equation*}
m_\delta(\tau) = \frac{\psi\lambda' e^{\tau}} {1 + \psi \lambda' (e^{\tau}-1)}.
\end{equation*}
The time transformation between $X_\psi^\delta$ and $Y$ is given by $g_3 = \rho_\delta$, with
\begin{align} \label{delta_trf}
&t = g_3(\delta) = \log(1 + \psi \lambda'(e^\delta - 1)), \\
&\delta = g_3^{-1}(t) = \rho_\delta^{-1}(t) = \log(1 + \frac{1}{\psi\lambda'}(e^t-1)), \label{delta_formula}
\end{align}


\subsection{Sampling method}

\citet{lambert} showed the following relationship between the two sampling scenarios when considered from a CPP perspective. Bernoulli sampled trees can be generated using the CPP formulation; that is, conditioning on a time of origin $T$, the event times are i.i.d.\ according to a specific density (as described in Section \ref{CPP_link}). For $n$-sampling, if we were to first generate a CPP tree with complete sampling (conditioned to have size at least $n$), and then choose $n$ lineages uniformly at random, then this would not have a CPP formulation \citep{lambertstadler}. However, the genealogy of such an $n$-sample  can be obtained by first drawing a sampling probability $\Psi=y$ from a specific improper prior, and then generating a Bernoulli CPP of size $n$ with sampling probability $y$. The improper prior has the form \citep[Theorem 3]{lambert}: 
\[
\frac{n(1-a)y^{n-1}}{(1-a(1-y))^{n+1}},
\]
where $a = P(H<T)$ is the probability that the random variable corresponding to event times (in the complete tree) takes a value less than the specified time of origin.

The underlying population (of the complete tree) growing to infinity can be seen to correspond to the time of origin of the complete process growing to infinity, and thus the probability $a = P(H<T)$ approaching 1. In this case, the improper prior on $\Psi$ tends to a point mass at $y = 0$. We do not explicitly condition on $T$, however this argument implies that the behaviour of the RRP for Bernoulli sampling with $\psi \to 0$, and for $n$-sampling when the underlying population grows to infinity, should be the same.


\subsection{Relationship between $X_\psi^\delta$ and $Y$ for small $\psi$} \label{X_psi_Y}

We examine the effect of the time rescaling between $X_\psi^\delta$ and  the time-reversed Yule rate 1 process $Y$, when $\psi \to 0$. In the following, we assume that $\lambda'$ is fixed and very small compared to $1/\psi$.

\begin{figure}[htbp!]
\centering
\includegraphics[width=0.9\textwidth]{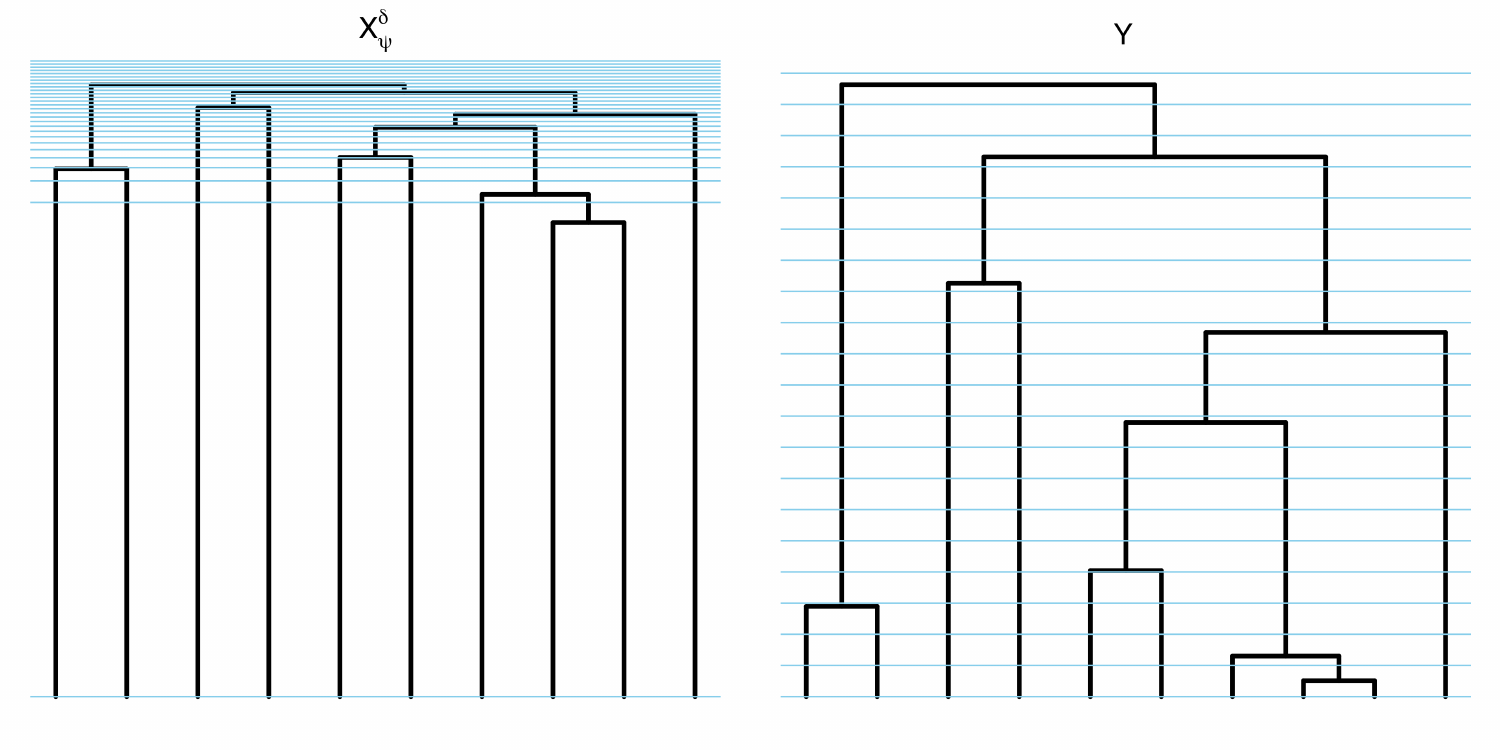}
\caption{Left: realisation of $X_\psi^\delta$ with $\psi = e^{-20}, \lambda'=2$. Right: same tree, rescaled in time units given by \eqref{delta_trf}. Intervals delineated by blue lines in the left panel are rescaled to intervals of equal length in the right panel.} \label{x_psi_v_y}
\end{figure}

Consider the time rescaling given by \eqref{delta_trf}: the process $X_\psi^\delta$ rescaled in units of $g_3(\delta)$ is  a time-reversed Yule rate 1 process. This rescaling is illustrated in Figure \ref{x_psi_v_y} for a small value of $\psi$; the left panel shows a realisation of $X_\psi^\delta$ for $n=10$. The right panel shows the same tree, but the intervals delineated by blue lines in the left panel are rescaled to intervals of equal length in the right panel. 

Using the identity $\log(1+x) = \log(x) + \log(1 + 1/x)$ and a Taylor expansion in $\psi\lambda'$ around 0, we obtain from \eqref{delta_formula}:
\begin{equation} \label{delta_time}
\delta - \log(\frac{1}{\psi\lambda'}) =  \log(e^t - 1) + \mathcal{O}(\psi\lambda').
\end{equation}
For small $t$, we have that $e^t-1 \approx t$ and the transformation behaves as $\delta- \log(1/(\psi\lambda')) \approx \log(t)$. For large $t$, we have $\log(e^t-1) \approx t$, so $\delta - \log(1/(\psi\lambda')) \approx t$. Thus, there are two time regimes, with a smooth transition between them.

\begin{figure}[htbp!]
\centering
\includegraphics[width=0.9\textwidth]{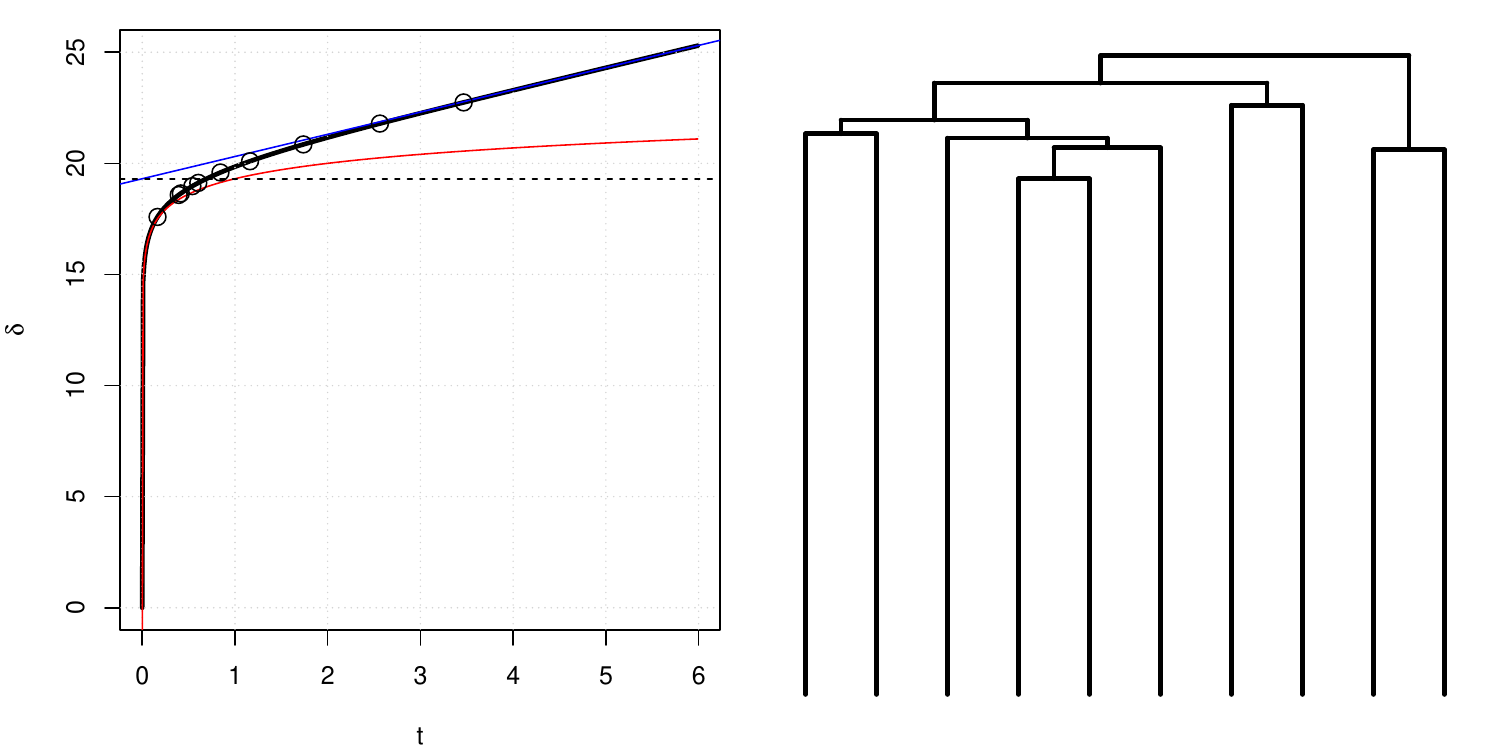}
\caption{Left: solid black line shows time rescaling between $\delta$ and $t$ (time units of $Y$). In blue: line $\delta =  t - \log(\psi \lambda')$. In red: curve $ \delta =  \log(t) - \log(\psi\lambda')$. Dashed black line shows $\delta = -\log(\psi \lambda')$. Dots show simulated event times. Right: tree corresponding to the simulated event times. Parameters used: $\psi = e^{-20}, \lambda'=2, \mu'=1$.} \label{t_v_delta}
\end{figure}

This can be understood as follows. Under Bernoulli sampling, the sample size $n$ is of order $\psi N$, where $N$ is the underlying population size in the complete tree. In the limit $\psi \to 0$, $N$ is therefore $\mathcal{O}(\psi^{-1})$; this is very large compared to $n$, and no coalescences happen for a very long time: the probability of going from $n$ to $n-1$ individuals in time $\tau$ is, from \eqref{transitions}:
\begin{equation*}
P_{n, n-1}(\tau) = \big(1 - e^{- \rho_\delta(\tau)} \big) \big(e^{{-\rho_\delta(\tau)}}\big)^{n-1},
\end{equation*}
where
\begin{equation*}
e^{-\rho_\delta(\tau)} = \frac{1}{ 1  + \psi \lambda' (e^{\tau} - 1)}
\end{equation*}
is the probability of no event happening. This is very close to 1 until $\tau$ grows to the order of $\log(1/\psi)$. 

With the time transformation above, a step of one unit of $t$ approximately corresponds to taking a time step of $\log(1/\psi)$ in units of $\delta$. At this point, $e^{-\rho_\delta(\tau)} \approx (1 + \lambda')^{-1}$ and the sample starts to coalesce. Then steps in $t$ become roughly equal to steps in $\delta$. In essence, we zoom back to a time when the underlying population was of order $n$, and then slow back down to linear time. 

Figure \ref{t_v_delta} shows an example of the time rescaling \eqref{delta_formula} for $\psi = e^{-20}, \, \lambda' = 2, \, \mu' = 1$. The left panel shows $\delta$ against $t$; the horizontal axis is the time scale of the $Y$ process, the vertical axis is the time scale of $X_\psi^\delta$. The red line shows the curve $\delta = \log(t) + \log(1/(\psi \lambda'))$; the blue line shows $\delta = t + \log(1/(\psi \lambda')) $. The circles indicate a set of simulated event times for a sample size $n=10$. For instance, the time to first event is $T_1 \sim \text{Exp}(n)$ on the horizontal axis; this is rescaled using \eqref{delta_formula} to get the corresponding time on the vertical axis. The right panel shows the corresponding RRP tree.

As $\psi \to 0$, $\log(1/(\psi\lambda')) \to \infty$, so the rescaled time of the first event in units of $\delta$ grows to infinity, and the reconstructed tree of $X_\psi^\delta$ becomes star-shaped. The terminal branches dominate the tree, but the inter-event times near the origin of the tree are still approximately exponentially distributed with rate depending on the remaining number of lineages, as the time rescaling for large $t$ is approximately linear.


\subsection{Density of inter-event times in the limit $\psi \to 0$} \label{interevent}

We now derive, analytically, the density of inter-event times, first for any $\psi \in (0,1]$, then for the limit $\psi \to 0$.

\begin{Theorem} \label{L1}
The density of waiting times $W_k = T_{k+1} - T_k$ between events $k$ and $k+1$, $k \in \{ 0, \hdots, n-1 \}$, for the RRP $X_\psi^\delta$ with $\psi \in (0,1]$, is:
\begin{multline} \label{Wk}
f^{\psi}_{W_k}(w) = \frac{(n-k)}{(n+1)} e^{-(n-k)w}  \Big[(n+1) {}_2F_1(n-k+1, n-k+1; n+1; (1 - \psi \lambda')(1-e^{-w}))\\
{}- (1 - \psi \lambda') (n-k+1) {}_2F_1(n-k+1, n-k+2; n+2; (1 - \psi \lambda')(1-e^{-w})) \Big],
\end{multline}
where ${}_2 F_1$ is the ordinary hypergeometric function. In the case of a critical branching process with birth and death rate $\lambda$ and RRP $Z_\psi^\alpha$, this becomes:
\begin{equation*}
\hat{f} ^{\psi}_{W_k}(v) =  \frac{(n-k+1) (n-k)}{n+1} \psi \lambda \cdot {}_2F_1 (n-k+1, n-k+2; n+2; -\psi \lambda v).
\end{equation*}
\end{Theorem}

Note that for $k=0$, $f^\psi_{W_0}(w)$ reduces to the density of the first event,  obtained by substituting $k=1$ in \eqref{kth_eventtime}.  We have the following case for $\psi \to 0$:

\begin{Corollary} \label{L2}The density of waiting times $W_k$ between events $k$ and $k+1$, $k \in \{1, \hdots, n-1 \}$, in the limit $\psi \to 0$, is:
\begin{equation} \label{fkw}
f^0_{W_k}(w) = \frac{k(n-k)}{n+1} e^{-(n-k)w} {}_2F_1(n-k+1, n-k+1; n+2; 1-e^{-w}).
\end{equation}
\end{Corollary}
This is not a density for $k=0$, i.e.\ for the waiting time to the first event; recall that for $\psi \to 0$ the first event time goes to infinity. 

Note that using the transformation \citep[p.64]{erdelyi}
\begin{equation*}
{}_2F_1(a, b; c; z) = (1-z)^{c - a - b} {}_2F_1(c-a, c-b; c; z),
\end{equation*}
the densities of the $k$-th and $(n-k)$-th waiting times are equal:
\begin{align}
f^{0}_{W_k}(w) &= \frac{k(n-k)}{n+1} e^{-(n-k)w} {}_2F_1(n-k+1, n-k+1; n+2; 1-e^{-w}) \nonumber \\
&= \frac{k(n-k)}{n+1} e^{-(n-k)w} e^{-(2k-n)w} {}_2F_1(k+1, k+1; n+2; 1-e^{-w}) \nonumber \\
&=  \frac{k(n-k)}{n+1} e^{-kw} {}_2F_1(k+1, k+1; n+2; 1-e^{-w}) \label{fk_forint}\\
&= f^0_{W_{n-k}}(w). \nonumber
\end{align}
This is an interesting property of the RRP tree in the limit. The inter-event times are symmetric, for instance the time it takes to go from $n-1$ to $n-2$ lineages, and the time it takes for the last lineage to die, have the same distribution. 

 To gain some insights into why this is true, consider the event times of the time-reversed Yule rate 1 process, which are distributed as the order statistics of $n$ exponential rate 1 random variables, say $X_1 \leq X_2 \leq \hdots \leq X_n$. The form of equation \eqref{delta_time} implies that in the limit $\psi \to 0$, the $k$-th event time $T_k$ can be obtained via the transformation $T_k = \log(1/(\psi \lambda')) + \log(e^{X_k} - 1)$. If $X \sim \text{Exp}(1)$, then $\log(e^X - 1)$ has the standard logistic distribution \citep{logistic_exp}. It thus follows that, in the limit, the shifted event time defined as $T'_k \defeq T_k - \log(1/(\psi \lambda'))$ is distributed as the $k$-th order statistic of $n$ draws from the standard logistic distribution, which has pdf 
\begin{equation} \label{logistic_pdf}
f^0_{T_{(1)}'} (\tau') = \frac{e^{\tau'}}{(1 + e^{\tau'})^2}.
\end{equation}
Note that this is equivalent to saying that $T_k$ is distributed as the $k$-th order statistic of $n$ draws from the logistic distribution with location parameter (mode) $\log(1/(\psi\lambda'))$ and scale 1. The same conclusion can also be reached by considering the coalescent point process density \eqref{death_one}, writing $\tau = \tau' + \log(1/(\psi \lambda')) $ and taking the limit $\psi \to 0$, which gives the density \eqref{logistic_pdf}.

The limiting density of $T_k'$ can also be obtained by applying the rescaling $\delta = (\lambda - \mu) \gamma$ and writing $\lambda' = \frac{\lambda}{\lambda - \mu}$ in the density \eqref{kth_eventtime},
\begin{equation*}
f^\psi_{T_k}(\tau) = {n \choose k} k \; \frac{\psi \lambda' e^\tau [\psi \lambda' (e^\tau - 1)]^{k-1}}{[1 + \psi \lambda' (e^\tau - 1)]^{n+1}},\nonumber
\end{equation*}
writing $T_k'  = T_k - \log(1/(\psi \lambda'))$ and taking the limit $\psi \to 0$ gives
\begin{equation}
f^0_{T_k'} ( \tau') = \lim_{\psi \to 0 }{n \choose k} k \; \frac{e^{\tau'} [e^{\tau'} - \psi \lambda' ]^{k-1}}{[1 + e^{\tau'} - \psi \lambda']^{n+1}}  = {n \choose k} k \; \frac{[e^{\tau'}]^{k}}{[1 + e^{\tau'}]^{n+1}}, \label{f_tkprime}
\end{equation}
which, again, is the density of the $k$-th order statistic for the standard logistic distribution.

\begin{figure}[htbp!]
\centering
\includegraphics[width=0.9\textwidth, trim={0 0 0 1.5cm}, clip]{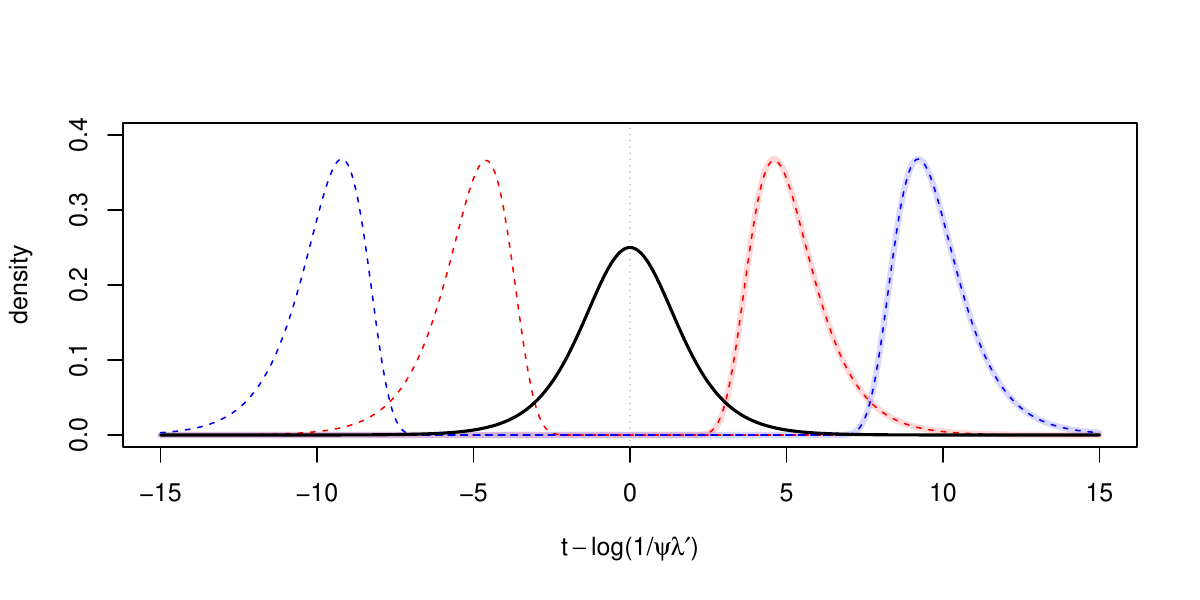}
\caption{ $x$-axis shows time shifted by $\log(1/(\psi \lambda'))$. Black solid line: standard logistic density \eqref{logistic_pdf}. Dashed lines: density \eqref{f_tkprime} of shifted time to first and last event for $n=100$ (red) and $n=10\ 000$ (blue). Faint solid lines: Gumbel density with parameters ($\log n, 1$) for $n=100$ (red) and $n=10\ 000$ (blue).} \label{hprime_density}
\end{figure}

As the logistic density \eqref{logistic_pdf} is symmetric around 0, the order statistics are also symmetric, with $T'_k \stackrel{d}{=} -T'_{n-k+1}$ \citep[pp.\ 26]{order_statistics}. This is illustrated in Figure \ref{hprime_density}: the black solid line shows the logistic density \eqref{logistic_pdf}, and the red (blue) dashed lines show the densities of the first and last event times for $n=100$ ($n=10\ 000$). Thus, the densities of the event times $T_k$ and $T_{n-k+1}$ are symmetric around $\log(1/(\psi \lambda'))$.

Moreover, as $T'_{k+1} \stackrel{d}{=} -T'_{n-k}$, this demonstrates that the inter-event times $W_k = T_{k+1} - T_{k} = T'_{k+1} - T'_k$ and $W_{n-k} = T_{n-k+1}-T_{n-k} = T'_{n-k+1}-T'_{n-k}$ are equal in distribution. The density derived in Corollary \ref{L2} is hence that of the gap between the $k$-th and $(k+1)$-th order statistic of $n$ standard logistic random variables. See for instance \citet[pp.\ 84]{logistic_paper}; their equation (4.1) gives the density of the gap between the $k$-th and $(k+1)$-th order statistics for the logistic distribution, which appears in very different form, but becomes the density in Corollary \ref{L2} after some algebra. We are not aware of a simpler expression for this particular density.

\begin{Corollary} \label{L1.1}
The distribution function of the waiting time $W_k$ between events $k$ and $k+1$, $k\in\{1,\dots,n-1\}$, with $\psi\to0$, is given by:
\begin{equation} \label{eq:distWk_psi0}
F^0_{W_k}(w) = 1 - e^{-kw} {}_2F_1(k, k+1; n+1; 1 - e^{-w}).
\end{equation}
\end{Corollary} 

 Another interesting property of this distribution is that it does not depend on the scaled birth rate $\lambda' = \frac{\lambda}{\lambda - \mu}$, as this parameter only appears as a factor in $\psi \lambda'$. In particular, we can take $\lambda' = 1 \implies \mu' = 0$. Thus, the inter-event times for the RRP $X_\psi^\delta$ have the same distributions as those of an incompletely sampled time-reversed Yule rate 1 process, in the limit $\psi \to 0$.

\subsection{Time to origin} \label{t_or_section}

We now consider the distribution of shifted time to origin $T_n'  = T_n -  \log(1/(\psi \lambda'))$ in the limit $\psi \to 0$. Integrating the density in \eqref{f_tkprime} for $k=n$, the distribution function of $T_n'$ is given by
\[
F^{0}_{T_n'} (\tau') = (1 + e^{-\tau'})^{-n}.
\]
As $n$ grows, the density of time to origin shifts to the right away from $\log(1/(\psi\lambda'))$, so with high probability $T_n'$ is much larger than 0. Figure \ref{hprime_density} demonstrates this visually with examples of the density of $T_n'$ for $n=100$ and $n=10\ 000$. Thus, for $n$ large enough, this justifies introducing the approximation $1 + e^{-\tau'} \approx \exp(e^{-\tau'})$, so the distribution of shifted time to origin can be approximated by
\[
\widetilde{F}^{0}_{T_n'} (\tau') = \left[ \exp(e^{-\tau'}) \right]^{-n} = \exp(-e^{-(\tau' - \log n)}).
\]
This is a Gumbel distribution with location parameter (mode) $\log n$ and scale parameter 1. Figure \ref{hprime_density} shows that this approximation provides a good fit, for $n=100$ and $n=10\ 000$.

This links to the results of \citet{burden}, who consider the diffusion limit (as the population size grows to infinity) of a near-critical Bienaym\'e-Galton-Watson process. \citet[Section 6]{burden} calculate numerically and plot the distribution of time to the MRCA, similarly shifted by the log of the population size at the time of sampling, and comment that as $n \to \infty$ this appears to converge to what looks like a Gumbel distribution. We have shown, analytically, that in the case of a supercritical birth-death process in the limit as $\psi \to 0$, the time to origin shifted by $\log(1/(\psi \lambda'))$ also converges to a Gumbel distribution, and in this case the location parameter depends on $n$.


\subsection{Exponential approximation of inter-event times}
Although Corollary \ref{L1.1} completely solves the question of what is the distribution of $W_k$ as $\psi\to 0$, the appearance of ${}_2F_1$ in \eqref{eq:distWk_psi0} somewhat obscures our insight into $W_k$. Here we show that these waiting times are well approximated by exponential distributions, so that the process is `almost' Markov.

Consider an exponential approximation to $f^{0}_{W_k}(w)$ with rate $k(n-k)/n$:
\begin{align} \label{exp_wait}
\widetilde{f}^{0}_{W_k} (w) &= \frac{k(n-k)}{n} \exp(-\frac{k(n-k)}{n} w), & \widetilde{F}^{0}_{W_k}(w) &= 1 - \exp(-\frac{k(n-k)}{n} w),
\end{align}
for $k \in \{ 1, \hdots, n-1 \}$. We have the following result concerning the accuracy of this approximation:
\begin{Proposition} \label{L3}
Suppose the waiting time distribution $W_k$, with distribution function \eqref{eq:distWk_psi0} for $\psi\to 0$, is approximated by an exponential distribution \eqref{exp_wait}. Then the approximation error is bounded, uniformly in $k$, in terms of Kolmogorov-Smirnov distance:
\[
\sup_w \abs{ F^{0}_{W_k}(w) - \widetilde{F}^{0}_{W_k}(w)} < \frac{1}{n}.
\]
\end{Proposition}

The density derived in Corollary \ref{L2} is nonintuitive, however this result shows that up to an error bounded by $1/n$, the distribution is actually approximately exponential. Note that the particular form of the exponential rate is such that $\widetilde f^{0}_{W_k}(w) = \widetilde f^{0}_{W_{n-k}}(w)$, so the symmetry between the $k$-th and $(n-k)$-th inter-event times is preserved in the approximation. Figure \ref{interevent_plot} shows an example of the (exact) density \eqref{Wk}, for $\psi=1$ on the left and very small $\psi$ on the right; dotted lines in the latter case show the exponential approximations \eqref{exp_wait}, demonstrating good agreement for $n=100$.

\begin{figure}[htbp!]
\centering
\includegraphics[width=0.9\textwidth]{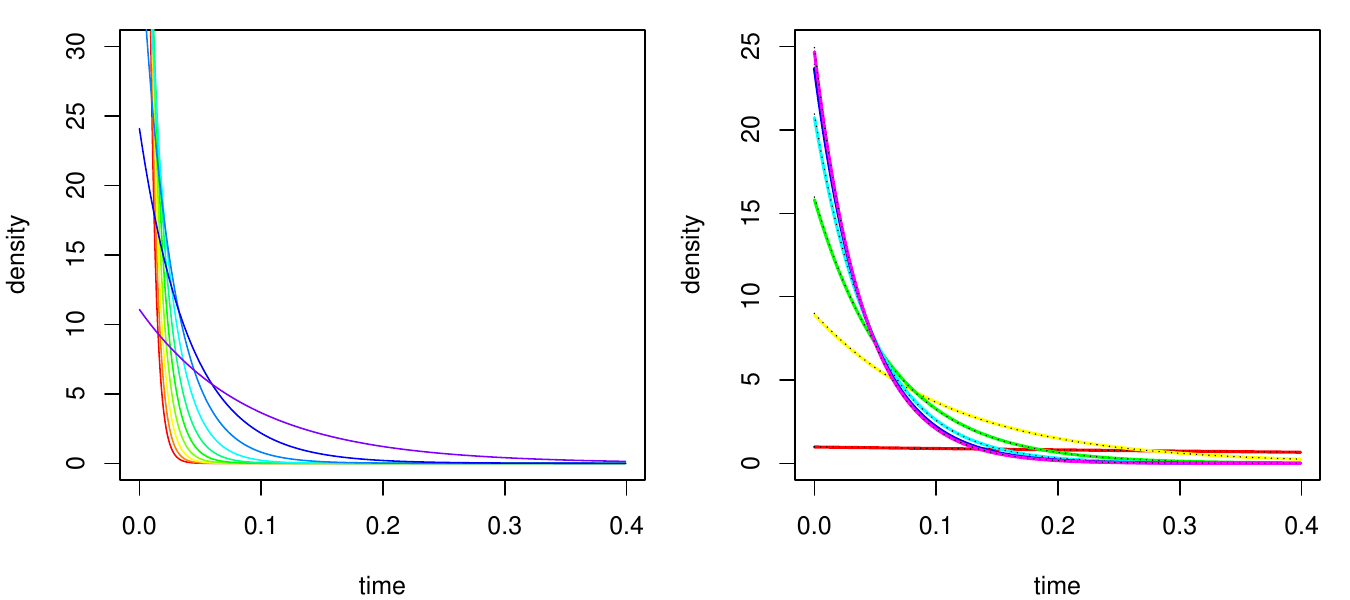}
\caption{ Inter-event time density, $n = 100$, $\lambda'=2$, $\mu'=1$. Left: with $\psi=1$, colours (red to purple) correspond to event numbers $k=0, 10, \hdots, 90$. Right: with $\psi = e^{-20}$, colours (red to purple) correspond to event numbers $k=1, 10, 20, \hdots, 50$; dotted lines show exponential approximation \eqref{exp_wait}.} \label{interevent_plot}
\end{figure}




\citet{wiuf_bd} gives results for the expectation of time to origin, and recursions for calculating the expectation of the other event times, for the RRP with Bernoulli sampling (not in the limit $\psi \to 0$). We can use these results to show that the expectation under the exponential approximation, being $n/(k(n-k))$, is \emph{exact} in the limit $\psi \to 0$ (for any $n$).

\begin{Proposition} \label{L4} The expectation of time to origin for $\psi \to 0$ is given by:
\begin{equation} \label{exp_t_or}
\mathbb{E}(T_{n}) = \log(\frac{1}{\psi \lambda'}) + \sum_{j=1}^{n-1} \frac{1}{j} + \mathcal{O}(\psi).
\end{equation}
\end{Proposition}
This is an illuminating result, as the expectation is split into two parts. The first is $\log(1/(\psi\lambda'))$, corresponding to the first time rescaling regime, as described in Section \ref{X_psi_Y}. Near 0, a small step in $t$ is equivalent to a step of order $\log(1/(\psi\lambda'))$ in units of $\delta$. The second part is equivalent to the expectation of a sum of $n-1$ exponential waiting times with rate being the remaining number of lineages, corresponding to the second time rescaling regime, which is approximately linear.  

This result agrees with the discussion in Section \ref{interevent}: recall that in the limit $\psi \to 0$, the shifted event time $T_k'$ is distributed as the $k$-th order statistic of $n$ standard logistic random variables, so $T_k = T_k' + \log(1/(\psi \lambda'))$ has expectation 
\begin{equation} \label{tk_expectation}
\sum_{j=1}^{k-1} \frac{1}{j} - \sum_{j=1}^{n-k} \frac{1}{j} + \log(\frac{1}{\psi \lambda'}),
\end{equation}
obtained by simplifying equation (4.8.6) in \citet[p.82]{order_statistics}. Setting $k=n$, this becomes \eqref{exp_t_or} up to the $\mathcal{O}(\psi)$ term. Notice also that using the Gumbel approximation for large $n$, as described in Section \ref{t_or_section}, gives the expectation of $T_n'$ as $\log n + \widetilde \gamma$ (where $\widetilde \gamma$ is the Euler-Mascheroni constant). This is the limit of the harmonic sum in \eqref{exp_t_or} as $n \to \infty$, so the expectations agree in this limit. 

\citet[Appendix D]{wiuf_bd} derives a recursion for the expectations of event times, which in our notation becomes: 
\begin{equation} \label{wiuf_recursion}
\mathbb{E}_n (T_{k}) = \frac{n}{n-k} \mathbb{E}_{n-1} (T_{k}) - \frac{k}{n-k} \mathbb{E}_n (T_{k+1}),
\end{equation}
 where $\mathbb{E}_n(T_k)$ denotes the expectation of the $k$-th event time if the sample is of size $n$ at time 0.  Using this and the expression for time to origin given by Proposition \ref{L4},  we obtain the following result:
\begin{Proposition} \label{Prop_Wk}
The expectation of waiting times between events is given by:
\begin{equation*} \label{wiuf_wait_time}
\mathbb{E}(W_k) = \mathbb{E}(T_{k+1}) - \mathbb{E}(T_{k}) = \frac{n}{k(n-k)} + \mathcal{O}(\psi).
\end{equation*}
\end{Proposition}
\noindent This agrees exactly with the expectation using the exponential approximation for $\psi \to 0$.  This also agrees, up to the $\mathcal{O}(\psi)$ term, with the expectation of $T_{k+1} - T_k$ obtained using \eqref{tk_expectation} in the limit $\psi \to 0$.


\section{Discussion} \label{discussion_section}

In this paper, we have demonstrated that viewing the RRP as an inhomogeneous pure-death process allows for relatively simple and intuitive derivations of its properties. The time rescaling approach allows for results derived for completely sampled RRPs to be transformed to those for incomplete sampling, using a simple change of variables, with no restrictions on the parameter values. Moreover, the time rescaling between the  time-reversed Yule rate 1 process  and the RRP can be used to simulate the RRP in a straight forward way, by simulating each event time sequentially. 

In the limit $\psi \to 0$, this rescaling can be decomposed into two timescales. The RRP tree becomes star-shaped, with terminal branch lengths tending to infinity, but inter-event times at the top of the tree are approximately exponential with a rate depending on $n$ and the event number. This has interesting implications for data analysis, as it suggests that the number of singleton mutations in a small sample from a very large population tends to infinity, but the number of shared mutations does not. Indeed, the recent paper of \citet{lambert_new} considers the expected frequency spectrum of mutations using a birth-death model with the infinite sites assumption. Although this is not explicitly discussed, the results of the simulations show that for small values of $\psi$, the expected number of singletons is orders of magnitude larger than that of mutations shared by multiple individuals. Taking the limit as $\psi \to 0$ in their equation (8), the expected number of singletons for $X_\psi^\delta$ grows to infinity, while for $k>1$
\[
\mathbb{E}\big[ S_n(k) \big] = \theta \, \frac{n+k-1}{k(k-1)},
\]
where $\theta$ is the mutation rate and $S_n(k)$ is the number of mutations with multiplicity $k$ in the sample of size $n$. In applying their method to cancer data, \citet{lambert_new} consider small values of $\psi$ with the population size being very large compared to the sample size---our results presented in Section \ref{psi_0_section} provide an insight into the properties of the genealogy in this case.

As can be seen from our results and related work, properties of the genealogy of a sample obtained from a population following a birth-death process are notably different from those arising under the coalescent,  particularly when the sample size is close to being of the same order as the population size.  The coalescent is widely used in statistical inference for intra-host viral and bacterial populations \citep[e.g.][]{dia:etal:2016}.   However, the choice of model should be appropriate to the relative scale of the biological application, and the individual-level population dynamics are arguably likely to be better modelled by a birth-death process. When considering the scenario of a small sample obtained from a very large population, the differences between the coalescent and the small-$\psi$ limit of the birth-death model will carry through to the resulting inference.  An important question is thus whether, for samples of viral or bacterial genetic sequencing data, birth-death models can provide better inference on the evolutionary dynamics of such populations. Answering this would require development of new methods for statistical inference that condition on the data and incorporate the natural processes governing such populations, such as high rates of mutation, recombination, and rapid demographic changes. This also presents interesting challenges in making full use of the increasingly rich sequencing data available for viral and bacterial infections.

\section*{Acknowledgements}
We thank two anonymous referees for their helpful comments. This work was supported by the OxWaSP CDT under the EPSRC grant EP/L016710/1, and by the Alan Turing Institute under the EPSRC grant EP/N510129/1.

\bibliography{Source_files_literature}{}
\bibliographystyle{Source_files_statsy}

\appendix

\section{Proofs} \label{proofs_appendix}


\subsection{Proof of Theorem \ref{L1}}

\begin{proof}
In the  time-reversed  Yule rate 1 process, the density of waiting times between the $k$-th and $k+1$-th event, $k = 0, \hdots, n-1$, conditional on $T_k = s$ is:
\begin{equation}
f_{W_k}(t | s) = (n-k)e^{-(n-k)((s+t)-s)} = (n-k)e^{-(n-k)t}. \label{yulefk}
\end{equation}
Using the time transformation \eqref{delta_trf}, in units of $\delta$ the waiting time is:
\[
w = \rho_\delta^{-1}(s+t) - \rho_\delta^{-1}(s) = \log(\frac{\psi \lambda' + e^{s+t}-1}{\psi \lambda' + e^s-1}).
\]
Rearranging, this gives:
\[
t = \log(e^{w} \left[ 1 - (1 - \psi \lambda')e^{-s} \right] + (1 - \psi \lambda') e^{-s}),
\]
and
\[
\frac{dt}{dw} = \frac{e^{w}(1 - (1 - \psi \lambda') e^{-s})}{e^{w}(1 - (1 - \psi \lambda')e^{-s}) + (1 - \psi \lambda')e^{-s}}.
\]
Thus, by using a change of variables in \eqref{yulefk} and writing $\phi = 1 - \psi \lambda'$:
\begin{align*}
f^\psi_{W_k} (w | s) &= (n-k) e^{w}  \big[1 - (1 - \psi \lambda')e^{-s} \big] \big[ e^{w}(1 - (1 - \psi \lambda')e^{-s}) + (1 - \psi \lambda')e^{-s}\big]^{-(n-k+1)} \\
&= (n-k) e^{w}  \big[1 - \phi e^{-s} \big] \big[ e^{w}(1 - \phi e^{-s}) + \phi e^{-s}\big]^{-(n-k+1)}.
\end{align*}
Since $s$ is the time of the $k$-th event in the  time-reversed  Yule rate 1 process, it has density given by \eqref{pd1_kthevent}:
\[
f_{T_k}(s) = {n \choose k-1} (n - k + 1) \big(1 - e^{-s} \big)^{k-1} \big(e^{-s} \big)^{n-k+1}.
\]
The marginal distribution of $W_k$ is thus
\begin{align*}
f^\psi_{W_k}(w) &= \int_0^\infty f^\psi_{W_k} (w | s) \, f_{T_k}(s) ds \\ 
&= {n \choose k-1} (n - k + 1) (n-k) \underbrace{\int_0^\infty e^{w} \frac{(1 - \phi e^{-s}) (1 - e^{-s})^{k-1} (e^{-s})^{n-k+1}}{ \big( e^{w}(1 - \phi e^{-s}) + \phi e^{-s}\big)^{n-k+1}} ds}_{(*)}.
\end{align*}
Integrating using the change of variables $u = e^{-s}$:
\begin{align*}
(*) &= e^{w} \int_0^1 \frac{(1 - \phi u)(1 - u)^{k-1} u^{n-k}}{ \big( e^{w}(1-\phi u) + \phi u\big)^{n-k+1}} du \nonumber \\
&= e^{-(n-k)w} \int_0^1 \frac{(1 - \phi u)(1 - u)^{k-1} u^{n-k}}{ \big( 1 - \phi u (1-e^{-w})\big)^{n-k+1}} du \\
&= e^{-(n-k)w} \int_0^1 \Bigg[ \frac{(1 - u)^{k-1} u^{n-k}}{ \big( 1 - \phi u (1-e^{-w})\big)^{n-k+1}} - \phi \frac{(1 - u)^{k-1} u^{n-k+1}}{ \big( 1 - \phi u (1-e^{-w})\big)^{n-k+1}} \Bigg] du.
\end{align*} 
Using the following identity for the ordinary hypergeometric function \citep[p.558]{abramowitz_stegun}:
\[
_2F_1(a, b, c, x) = \frac{\Gamma(c)}{\Gamma(c-a) \Gamma(a)} \int_0^1 \frac{(1 - t)^{c-a-1} t^{a-1}}{ (1 - xt)^{b}} dt,
\]
we obtain
\begin{align*}
(*) = e^{-(n-k)w} \frac{(k-1)!(n-k)!}{(n+1)!} \Big[&(n+1) {}_2F_1(n-k+1, n-k+1; n+1; \phi(1-e^{-w})) \\
&{}- \phi (n-k+1) {}_2F_1(n-k+1, n-k+2; n+2; \phi(1-e^{-w})) \Big].
\end{align*}
Thus,
\begin{multline*} \label{fkt}
f^\psi_{W_k}(w) =  \frac{(n-k)}{(n+1)} e^{-(n-k)w}  \Big[(n+1) {}_2F_1(n-k+1, n-k+1; n+1; (1 - \psi \lambda')(1-e^{-w})) \\ 
{}- (1 - \psi \lambda') (n-k+1) {}_2F_1(n-k+1, n-k+2; n+2; (1 - \psi \lambda')(1-e^{-w})) \Big].
\end{multline*}

For the RRP of a critical branching process, $Z_\psi^\alpha$, the derivation is very similar. Using instead the time transformation 
\[
v = \rho_\alpha^{-1}(s+t) - \rho_\alpha^{-1}(s) = \frac{1}{\psi\lambda}\Big[ e^{s+t} - 1 - e^s + 1 \Big] =  \frac{1}{\psi\lambda}e^s (e^t-1)
\]
and following the same steps, we obtain
\begin{align*}
\hat{f}^\psi_{W_k}(v) =  & \frac{(n-k+1) (n-k)}{n+1} \psi \lambda \cdot {}_2F_1 (n-k+1, n-k+2; n+2; -\psi \lambda v).
\end{align*}
\end{proof}


\subsection{Proof of Corollary \ref{L2}}

\begin{proof}
Substituting $\psi = 0$ into \eqref{Wk}:
\begin{multline*} \label{fkt}
f^0_{W_k}(w) =  \frac{(n-k)}{(n+1)} e^{-(n-k)w}  \Big[(n+1) {}_2F_1(n-k+1, n-k+1; n+1; 1-e^{-w}) \\ 
{}- (n-k+1) {}_2F_1(n-k+1, n-k+2; n+2; (1-e^{-w})) \Big].
\end{multline*}
 Identity (15.2.16) of \citet[p.\ 558]{abramowitz_stegun} gives:
\begin{equation} \label{f1}
ac(1-z) \;{}_2 F_1 (a+1,b;c;z) = c[a-(c-b)z] \;{}_2 F_1 (a,b;c;z) + (c-a)(c-b)z \;{}_2 F_1 (a,b;c+1;z)
\end{equation}
Substituting $a+1$ instead of $a$ in identity (15.2.20) of \citet[p.\ 558]{abramowitz_stegun} gives:
\begin{equation} \label{f2}
c(1-z) \;{}_2 F_1 (a+1,b;c;z) = c  \;{}_2 F_1 (a, b; c ;z)  - (c-b)z  \;{}_2 F_1 (a+1,b;c+1;z).
\end{equation}
Multiplying \eqref{f2} by $a$, equating with \eqref{f1} and simplifying gives: 
\[
c \;{}_2 F_1 (a,b; c;z) - a  \;{}_2 F_1 (b, a+1; c+1; z) = (c-a)  \;{}_2 F_1 (a,b; c+1;z).
\]

Thus, we obtain
\begin{equation*}
f^0_{W_k}(w) = \frac{k(n-k)}{(n+1)} e^{-(n-k)w} {}_2F_1(n-k+1, n-k+1; n+2; 1-e^{-w}).
\end{equation*}
\end{proof}


\subsection{Proof of Corollary \ref{L1.1}}

\begin{proof}
By integrating the density in \eqref{fk_forint}:
\begin{align*}
F^0_{W_k}(w) &= \frac{k(n-k)}{n+1} \int_0^w  e^{-ku} {}_2F_1(k+1, k+1; n+2; 1-e^{-u}) du \nonumber \\
&= \frac{k(n-k)}{n+1} \int_0^w e^{-u} e^{-(k-1)u} {}_2F_1(k+1, k+1; n+2; 1-e^{-u}) du \nonumber\\
&= \frac{k(n-k)}{n+1} \int_0^{1 - e^{-w}}  (1-z)^{k-1} {}_2F_1(k+1, k+1; n+2; z) dz \nonumber\\
&= \frac{k(n-k)}{n+1} \Big[ -\frac{(1-z)^{k} (n+1)}{k(n-k)} \;{}_2F_1(k, k+1; n+1; z) \Big]_0^{1 - e^{-w}} \nonumber \\
&= 1 - e^{-kw} {}_2F_1(k, k+1; n+1; 1 - e^{-w}),
\end{align*}
having used the substitution $z = 1 - e^{-u}$, and the identity  \citep[p.102, eq.\ (25) with $n=1$]{erdelyi} 
\begin{equation*}
\int^z (1 - x)^{a-2} {}_2F_1 (a, b, c, x) \; dx = \frac{c-1}{(a-1)(b-c+1)}(1 - z)^{a-1} {}_2F_1 (a-1, b, c-1, z).
\end{equation*}
\end{proof}


\subsection{Proof of Proposition \ref{L3}}

\begin{proof}
Noting that 
\[
e^{-kw} = \exp(-\frac{k(n-k)}{n} w) \cdot \exp(-\frac{k^2}{n} w),
\]
we have:
\begin{align}
\big| \widetilde{F}^0_{W_k} (w) - F^0_{W_k}(w) \big| &= \Bigg| 1 - e^{-kw} {}_2F_1(k, k+1; n+1; 1 - e^{-w}) - 1 + \exp(-\frac{k(n-k)}{n} w) \Bigg| \nonumber \\
&=  \exp(-\frac{k(n-k)}{n} w) \cdot \Bigg|  \underbrace{  \exp(-\frac{k^2}{n} w) {}_2F_1(k, k+1; n+1; 1 - e^{-w})}_{=:\,h(w)} - 1 \Bigg| . \label{bound}
\end{align}
We need to obtain an upper bound on the maximum of this distance. The first exponential term decays rapidly to 0, while $h(0)=1$ and $h$ initially increases; the global maximum of $h$ occurs near $w=0$, where $h(w)-1 \geq 0$. We first obtain an upper bound on $h(w)-1$, and then use this to obtain an upper bound on \eqref{bound}. Using the mean value theorem (or, equivalently, Taylor's theorem to first order):
\begin{equation*} \label{g}
h(w) = h(0) + w h'(c) = 1 + w h'(c)
\end{equation*}
for some $c \in (0, w)$, with
\begin{align*}
h'(c) = &{}-\frac{k^2}{n}\exp(-\frac{k^2}{n} c) {}_2F_1(k, k+1; n+1; 1 - e^{-c}) \nonumber \\
&{}+ \exp(-\frac{k^2 + n}{n} c).\frac{k(k+1)}{n+1} {}_2F_1(k+1, k+2; n+2; 1 - e^{-c}).
\end{align*}
Differentiating again and considering the sign of the second derivative, we find that $h''(0)<0$, so $h'$ has a maximum at $c=0$; $h'$ has no other extrema before it reaches 0. We have:
\begin{align*}
h'(0) = -\frac{k^2}{n} + \frac{k(k+1)}{n+1} = \frac{k(n-k)}{n(n+1)},
\end{align*}
so an upper bound on $h(w)-1$ is given by 
\[
h(w)-1 \leq \frac{k(n-k)}{n(n+1)} w.
\]
Substituting this into \eqref{bound}:
\begin{align}
\big| \widetilde{F}^0_{W_k} (w) - F^0_{W_k}(w) \big| &\leq  \exp(-\frac{k(n-k)}{n} w) \cdot \big( h(w) - 1 \big) \nonumber \\
& \leq \exp(-\frac{k(n-k)}{n} w) \cdot \frac{k(n-k)}{n(n+1)} w. \label{a17}
\end{align}
This attains the maximum at $\hat{w} = \frac{n}{k(n-k)}$. Substituting this into \eqref{a17}, we obtain the bound:
\[
\big| \widetilde{F}^0_{W_k} (w) - F^0_{W_k}(w) \big| \leq \frac{1}{e(n+1)} < \frac{1}{n}.
\]
The approximation error is thus bounded by $\frac{1}{n}$.
\end{proof}


\subsection{Proof of Proposition \ref{L4}}

\begin{proof}
\citet[Appendix F]{wiuf_bd} derives an expression for the expectation of time to origin, which in our notation is:
\begin{equation}
\mathbb{E}(T_n) = \log(\frac{1}{\psi \lambda'}) + \sum_{i=1}^n \frac{1}{i}  - \sum_{i=1}^n \frac{1}{i}\frac{1}{\Big(1 - \frac{1}{\psi \lambda'}\Big)^{n-i}} - \frac{1}{(1 - \frac{1}{\psi \lambda'})^n} \log(\frac{1}{\psi \lambda'}). \label{eq:wiufF}
\end{equation}
The third term is:
\[
 \sum_{i=1}^n \frac{1}{i}\frac{1}{\Big(1 - \frac{1}{\psi \lambda'}\Big)^{n-i}} = \frac{1}{n} + \frac{1}{n-1} \frac{1}{1 - \frac{1}{\psi \lambda'}} + \frac{1}{n-2} \Bigg( \frac{1}{1 - \frac{1}{\psi \lambda'}}\Bigg)^2 + \hdots = \frac{1}{n} + \mathcal{O}(\psi).
\]
 The fourth term in \eqref{eq:wiufF} is:
\begin{align*}
\frac{1}{(1 - \frac{1}{\psi \lambda'})^n} \log(\frac{1}{\psi \lambda'}) &= (- \psi \lambda')^n (1 - \psi \lambda') ^{-n} \log(\frac{1}{\psi \lambda'}) \\
&=  -  (- \psi \lambda')^n [1 + \mathcal{O} (\psi \lambda')] \log(\psi \lambda') \\
&= \mathcal{O} ((\psi \lambda')^n \log(\psi \lambda)),
\end{align*}
which is $\mathcal{O}(\psi)$ for $n>1$.  In the limit $\psi \to 0$, we thus have 
\begin{equation*}
\mathbb{E}(T_n) = \log(\frac{1}{\psi \lambda'}) + \sum_{i=1}^{n-1} \frac{1}{i} + \mathcal{O}(\psi).
\end{equation*}
\end{proof}

\subsection{Proof of Proposition \ref{Prop_Wk}}
\begin{proof}
From Proposition \ref{L4}, the expectation of time to origin for a sample of size $n$ is:
\begin{equation*}
\mathbb{E}_n(T_{n}) = \log(\frac{1}{\psi \lambda'}) + \sum_{j=1}^{n-1} \frac{1}{j}+ \mathcal{O}(\psi),
\end{equation*}
which also implies that, for a sample of size $n-1$,
\begin{equation*}
\mathbb{E}_{n-1}(T_{n-1}) = \log(\frac{1}{\psi \lambda'}) + \sum_{j=1}^{n-2} \frac{1}{j} + \mathcal{O}(\psi).
\end{equation*}
We proceed by induction on the event number $k$, to show that 
\begin{equation} \label{ind}
\mathbb{E}_n(T_k) = \log(\frac{1}{\psi \lambda'}) + \sum_{j=1}^{n-1}\frac{1}{j}  - \sum_{j=k}^{n-1} \frac{n}{j(n-j)} + \mathcal{O}(\psi).
\end{equation}
This holds for event number $k=n-1$, as using \eqref{wiuf_recursion}:
\begin{align*}
\mathbb{E}_n(T_{n-1}) &= n \mathbb{E}_{n-1}(T_{n-1}) - (n-1) \mathbb{E}_n(T_{n}) \\
&= n \left( \log(\frac{1}{\psi \lambda'}) + \sum_{j=1}^{n-1} \frac{1}{j} - \frac{1}{n-1} \right) - (n-1) \left(\log(\frac{1}{\psi \lambda'}) + \sum_{j=1}^{n-1} \frac{1}{j} \right) + \mathcal{O}(\psi)\\
&= \log(\frac{1}{\psi \lambda'}) + \sum_{j=1}^{n-1} \frac{1}{j} - \frac{n}{n-1} + \mathcal{O}(\psi).
\end{align*}
Suppose that \eqref{ind} holds for some $k = n-i$, $i \in \{ 1, \hdots, n-1 \}$:
\begin{equation*}
\mathbb{E}_n(T_{n-i}) = \log(\frac{1}{\psi \lambda'}) + \sum_{j=1}^{n-1}\frac{1}{j}  - \sum_{j=1}^{i} \frac{n}{j(n-j)} + \mathcal{O}(\psi),
\end{equation*}
and so, equivalently, for $n-1$ lineages:
\begin{equation*}
\mathbb{E}_{n-1}(T_{n-i-1}) = \log(\frac{1}{\psi \lambda'}) + \sum_{j=1}^{n-2}\frac{1}{j}  - \sum_{j=1}^{i} \frac{n-1}{j(n-j-1)} + \mathcal{O}(\psi).
\end{equation*}
Then:
\begin{align*}
\mathbb{E}_n(T_{n-i-1}) &= \frac{n}{i+1}\mathbb{E}_{n-1}(T_{n-i-1}) - \frac{n-i-1}{i+1} \mathbb{E}_n(T_{n-i}) \\
&= \underbrace{\log(\frac{1}{\psi \lambda'}) + \sum_{j=1}^{n-1}\frac{1}{j} + \mathcal{O}(\psi)}_{(*)} - \frac{n}{i+1} \left[  \frac{1}{n-1}  + \sum_{j=1}^{i} \frac{n-1}{j(n-j-1)}  - (n-i-1)\sum_{j=1}^{i} \frac{1}{j(n-j)}  \right]  \\
&= (*) - \frac{n}{i+1} \left[  \frac{1}{n-1}  +  \sum_{j=1}^{i} \left( \frac{1}{j}  + \frac{1}{n-j-1} \right) - \frac{(n-i-1)}{n}\sum_{j=1}^{i} \left( \frac{1}{j} + \frac{1}{n-j} \right)  \right] \\
&= (*)  - \frac{1}{i+1} \left[  \frac{n}{n-1}  +  (i+1) \sum_{j=1}^{i}\frac{1}{j}  +  n \sum_{j=1}^{i} \frac{1}{n-j-1}  - (n-i-1)\sum_{j=1}^{i}  \frac{1}{n-j}  \right] \\
&= (*)  - \frac{1}{i+1} \left[   (i+1) \sum_{j=1}^{i}\frac{1}{j}  +  n \sum_{j=2}^{i} \frac{1}{n-j} + \frac{n}{n-i-1} - (n-i-1)\sum_{j=2}^{i}  \frac{1}{n-j}  + \frac{i+1}{n-1} \right]  \\
&= (*)  - \frac{1}{i+1} \left[ (i+1) \sum_{j=1}^{i}\frac{1}{j}  +  (i+1) \sum_{j=1}^{i} \frac{1}{n-j} + \frac{(n-i-1)+(i+1)}{n-i-1} \right]  \\
&=(*)  - \frac{1}{i+1} \left[ (i+1) \sum_{j=1}^{i+1}\frac{1}{j}  +  (i+1) \sum_{j=1}^{i+1} \frac{1}{n-j}\right] \\
 &= \log(\frac{1}{\psi \lambda'}) + \sum_{j=1}^{n-1}\frac{1}{j} - \sum_{j=1}^{i+1} \frac{n}{j(n-j)} + \mathcal{O}(\psi).
\end{align*}
Thus,
\begin{equation*}
\mathbb{E}_n(T_k) = \log(\frac{1}{\psi \lambda'}) + \sum_{j=1}^{n-1}\frac{1}{j}  - \sum_{j=1}^{n-k} \frac{n}{j(n-j)} + \mathcal{O}(\psi) = \log(\frac{1}{\psi \lambda'}) + \sum_{j=1}^{n-1}\frac{1}{j}  - \sum_{j=k}^{n-1} \frac{n}{j(n-j)} + \mathcal{O}(\psi),
\end{equation*}
and so
\[
\mathbb{E}(W_k) = \mathbb{E}(T_{k+1}) -   \mathbb{E}(T_{k}) = \frac{n}{k(n-k)} + \mathcal{O}(\psi).
\]
\end{proof}


\newgeometry{left=2cm, bottom=2cm, top=2cm}
\newpage
\thispagestyle{empty}
\begin{landscape}

\section{Summary of RRPs} \label{appB}

\bgroup
\def\arraystretch{4}
\begin{table}[!htbp]
\centering
\begin{tabular}{| p{3cm} | C || C || C | C | C | C |   }
\hline
RRP & Y & Z_\psi^\alpha & X_1^\beta& X_\psi^\gamma& X_\psi^\delta\\ \hline

Time variable &  t & \alpha = \frac{1}{\lambda \psi}(e^t-1)& \beta = \frac{1}{\lambda - \mu} \log(1 + \frac{\lambda - \mu}{\lambda} (e^t-1))  &  \gamma = \frac{1}{\lambda - \mu} \log (1 + \frac{1}{\psi} (e^{(\lambda - \mu)\beta} - 1)) & \delta = (\lambda - \mu) \gamma \\  

        && t = \log(1 +  \psi \lambda \alpha)& t = \log(1 + \frac{ \lambda}{\lambda - \mu} \Big( e^{(\lambda - \mu)\beta} - 1 \Big)) & \beta =  \frac{1}{\lambda - \mu} \log (1 + \psi (e^{(\lambda - \mu)\gamma} - 1)) & \gamma = \frac{1}{\lambda - \mu} \delta \\ \hline

Corresponding complete process & \text{Yule(1)} & \text{CBP}(\lambda, \psi) & \text{BDP$(\lambda, \mu, 1)$} & \text{BDP$(\lambda, \mu, \psi)$} & \text{BDP$(\lambda', \mu', \psi)$ with } \lambda' = \frac{\lambda}{\lambda - \mu}, \mu' = \frac{\mu}{\lambda - \mu}  \\ \hline                                                 
                                                 
$m$ (death rate of the RRP, per lineage)            &  1 &  \frac{\psi \lambda}{ 1+ \psi \lambda \alpha} & \frac{\lambda e^{(\lambda - \mu)\beta}}{1 + \frac{ \lambda}{\lambda - \mu} (e^{(\lambda - \mu) \beta} - 1)} &  \frac{\psi \lambda e^{(\lambda - \mu)\gamma}}{1 + \frac{\psi \lambda}{\lambda - \mu} (e^{(\lambda - \mu) \gamma} - 1)}  &  \frac{\psi \lambda' e^{\delta}}{1 + \psi \lambda' (e^{\delta} - 1)} \\ \hline

$\rho = \int m $  &  t  &  \log (1 + \psi \lambda \alpha)& \log(1 + \frac{ \lambda}{\lambda - \mu} \Big( e^{(\lambda - \mu)\beta} - 1 \Big)) &  \log(1 + \frac{\psi \lambda}{\lambda - \mu} \Big( e^{(\lambda - \mu)\gamma} - 1 \Big)) &  \log(1 + \psi \lambda' \Big( e^{\delta} - 1 \Big))  \\ \hline

$e^{-\rho}$  & e^{-t} &  \frac{1}{1 + \psi \lambda \alpha}&\frac{1}{1 + \frac{ \lambda}{\lambda - \mu} \Big( e^{(\lambda - \mu)\beta} - 1 \Big)}  & \frac{1}{1 + \frac{\psi \lambda}{\lambda - \mu} \Big( e^{(\lambda - \mu)\gamma} - 1 \Big)} & \frac{1}{1 + \psi \lambda' \Big( e^{\delta} - 1 \Big)} \\ \hline

\end{tabular}
\end{table}

\egroup
\end{landscape}
\restoregeometry

\end{document}